\documentclass[nonatbib]{elsarticle}
\usepackage{xspace,amsmath,amsfonts,amssymb,bbm,ifthen,mathtools}
\usepackage[hidelinks,pdftex]{hyperref}
\usepackage[heavycircles]{stmaryrd}
\usepackage{bbold}
\usepackage[sort, noadjust]{cite}
\usepackage{wasysym}
\usepackage{enumitem}
\usepackage{tikz}
\usepackage{tikz-cd}
\usetikzlibrary{arrows}
\tikzset{->, auto, >=latex', label distance=-.25cm, shorten <=-3pt,
  shorten >=-3pt}
\tikzset{every node={font=\footnotesize}}

\newdefinition{remark}{Remark}
\newdefinition{definition}{Definition}
\newdefinition{example}{Example}
\newdefinition{problem}{Problem}
\newtheorem{theorem}{Theorem}

\newtheorem{lemma}[theorem]{Lemma}
\newtheorem{proposition}[theorem]{Proposition}
\newtheorem{corollary}[theorem]{Corollary}
\newproof{proof}{Proof}

\newcommand*\N{{\textup{\textsf{N}}}}
\newcommand*\twotwo{{\textup{\textsf{2\texttt{+}2}}}}
\newcommand*\NN{{\textup{\textsf{N\hspace{-.65ex}N}}}}
\newcommand*\NPLUS{{\textup{\textsf{M}}}}
\newcommand*\NMINUS{{\textup{\textsf{W}}}}
\newcommand*\TC{{\textup{\textsf{3C}}}}
\newcommand*\LN{{\textup{\textsf{LN}}}}

\newcommand*\closureop[1]{\mathord{#1}}
\newcommand*\closure[2]{\closureop{#1}#2}

\newcommand*\intpt{\text{\tiny {\Circle}}}
\newcommand*\inpt{\tiny {\ensuremath \RIGHTCIRCLE}}
\newcommand*\outpt{\tiny {\ensuremath \LEFTCIRCLE}}
\newcommand*\rintpt{{\raisebox{1pt}{\textup{\intpt}}}}
\newcommand*\rinpt{\mathop{\raisebox{1pt}{\inpt}}}
\newcommand*\routpt{\mathop{\raisebox{1pt}{\outpt}}}
\newcommand*\riopt{\mathop{\raisebox{1pt}{\inpt\hspace{-1.6pt}\outpt}}}
\newcommand*\rtone{\mathop{\raisebox{.5pt}{\tiny $1$}}}
\newcommand*\rti{\mathop{\raisebox{.5pt}{\tiny $i$}}}
\newcommand*\rtj{\mathop{\raisebox{.5pt}{\tiny $j$}}}

\newcommand*\ie{\textit{i.e.},\xspace}
\DeclareRobustCommand{\single}[1]{%
  \ifthenelse{\equal{#1}{4}}{%
    \mcal{S}%
  }{%
    \ifthenelse{\equal{#1}{1}}{%
      \mcal T%
    }{%
      \errorUndefinedArgument%
    }}}

\newcommand*\incomp{\mathrel{\|}}
\newcommand*\up{\mathord{\uparrow}}
\newcommand*\down{\mathord{\downarrow}}
\newcommand*\intord{\dashrightarrow}
\newcommand*\mcal[1]{\mathcal{#1}}
\newcommand*\from{\leftarrow}
\newcommand*\To{\Rightarrow}
\newcommand*\glue{\mathbin{\triangleright}}
\newcommand*\op{\textup{op}}
\let\po\vec
\newcommand*\poo[1]{\overrightarrow{#1}}
\newcommand*\pomser{\mathbin{*}}
\newcommand*\Nat{\mathbbm{N}}
\newcommand*\Real{\mathbbm{R}}
\newcommand*\subsu{\preceq}
\newcommand*\id{\textup{\textsf{id}}}
\newcommand*\dom{\textit{dom}}
\newcommand*\cod{\textit{cod}}

\newcommand*\cat[1]{\text{\textup{\textsf{#1}}}}
\newcommand*\Pos{\cat{Pos}}
\newcommand*\iPos{\cat{iPos}}

\begin{document}

\journal{Information and Computation}

\begin{frontmatter}

\title{Posets with Interfaces as a Model for Concurrency}

\author[1]{Uli Fahrenberg\corref{cor1}}
\ead[1]{uli@lrde.epita.fr}

\author[2]{Christian Johansen} 
\ead[2]{christian@johansenresearch.info}

\author[3]{Georg Struth} 
\ead[3]{g.struth@sheffield.ac.uk}

\author[4]{Krzysztof Ziemia{\'n}ski}
\ead[4]{ziemians@mimuw.edu.pl}

\address[1]{EPITA Research and Development Laboratory (LRDE), France}
\address[2]{Norwegian University of Science and Technology, Norway} 
\address[3]{University of Sheffield, UK }
\address[4]{University of Warsaw, Poland}

\cortext[cor1]{Corresponding author.}

\date{}

\begin{abstract}
  We introduce posets with interfaces (iposets) and generalise their
  standard serial composition to a new gluing composition.  In the
  partial order semantics of concurrency, interfaces and gluing allow
  modelling events that extend in time and across components.
  Alternatively, taking a decompositional view, interfaces allow
  cutting through events, while serial composition may only cut
  through edges of a poset.  We show that iposets under gluing
  composition form a category, which generalises the monoid of posets
  under serial composition up to isomorphism. They form a $2$-category
  when a subsumption order and a lax tensor in the form of a
  non-commutative parallel composition are added, which generalises
  the interchange monoids used for modelling series-parallel
  posets. We also study the gluing-parallel hierarchy of iposets,
  which generalises the standard series-parallel one.  The class of
  gluing-parallel iposets contains that of series-parallel posets and
  the class of interval orders, which are well studied in concurrency
  theory, too.  We also show that it is strictly contained in the
  class of all iposets by identifying several forbidden substructures.
 
\end{abstract}

\end{frontmatter}

\section{Introduction}

Our general motivation for studying posets with interfaces (iposets)
comes from concurrent Kleene algebra
\cite{DBLP:journals/jlp/HoareMSW11}. There, such structures have been
proposed as semantics for concurrent programs, because the
compositionality inherent to standard partial order models limits
their applicability (see \cite{DBLP:books/mc/21/Struth21} for a
survey). Our particular conceptual choices of interfaces and
operations on iposets are motivated by languages of higher-dimensional
automata \cite{Hdalang}. Our gluing composition is meant to match the
gluing of such automata along faces, our non-commutative parallel
composition to correspond to their tensor product. Higher-dimensional
automata generalise many other models of concurrency
\cite{DBLP:journals/tcs/Glabbeek06}, which makes the study of their
languages relevant to concurrent Kleene algebra.

Nevertheless, this article is not about Kleene algebras or
higher-dimensional automata per se. Instead we study the basic
order-theoretic and algebraic properties of iposets, explore their
categorical structure and investigate their alternation hierarchy
under parallel and gluing composition.

Our starting point are partial order semantics of concurrency
\cite{DBLP:books/sp/Vogler92}, where points of finite posets model
events of concurrent systems and partial order relations model
temporal precedences or causal dependencies between events, as well as
events that occur in parallel or independently. Events are often
labelled with the actions of a concurrent system. One then tends to
forget individual events and model only structural properties of
labelled posets up to isomorphism. Isomorphism classes of labelled
posets are known as \emph{partial words} or \emph{pomsets}
\cite{DBLP:journals/ipl/Winkowski77, DBLP:journals/fuin/Grabowski81,
  Pratt86pomsets, DBLP:journals/tcs/Gischer88}. Pomsets are usually
equipped with two compositions. Intuitively, a parallel composition
lays out one pomset above another. A serial one lays out one pomset to
the left of another and extends their orders so that each event in the
left pomset precedes every event in the right one.  The
series-parallel pomsets are then generated by the empty pomset and all
one-event pomsets, and closed under serial and parallel
compositions. It is a strict subclass of the class of all pomsets,
because it excludes precisely those pomsets that contain an induced
subposet of the shape
\begin{equation*}
  \N=\! \vcenter{\hbox{%
      \begin{tikzpicture}[y=.5cm]
        \node (0) at (0,0) {\intpt};
        \node (1) at (0,-1) {\intpt};
        \node (2) at (1,0) {\intpt};
        \node (3) at (1,-1) {\intpt};
        \path (0) edge (2) (1) edge (2) (1) edge (3);
      \end{tikzpicture}}} .
\end{equation*}
Finally, a subsumption order compares pomsets with the same
set of events in terms of their precedence relations
\cite{DBLP:journals/fuin/Grabowski81, DBLP:journals/tcs/Gischer88}.
For so-called \emph{concurrent monoids}, that is, ordered double
monoids with shared units in which the serial composition $\cdot$ and
the commutative parallel composition $\times$ interact via the lax
interchange law
\begin{equation*}
  (x_1\times x_2)\cdot (y_1\times y_2)\le (x_1\cdot y_1)\times
  (x_2\cdot y_2),
\end{equation*}
the series-parallel pomsets over the set $A$ of labels and with 
subsumption order $\le$ are then freely generated by $A$ in this class of
algebras~\cite{DBLP:journals/tcs/Gischer88,DBLP:journals/mscs/BloomE97}.

Here we adapt these notions and results to iposets. As we are mainly
interested in the order structure, we ignore labels and tacitly work
with isomorphism classes of finite posets.

One motivation for our use of interfaces comes from modeling events
with duration or structure (by contrast to interleaving models of
concurrency where they are instantaneous and indivisible).  In the
standard partial order semantics of concurrency we may interpret the
order on events as temporal precedence.  A concurrent system can then
be decomposed into two parallel parts if and only if these are
temporally disconnected.  Similarly, it can be serially decomposed if
and only if all events in the first component precede all events in
the second.  In iposets, more generally, we allow serial
decompositions to \emph{cut through events} as well, so that the
extend from the first component into the second. Such events then form
the target interface of the first component and the source interface
of the second.  See \cite{Pratt03trans_cancel, conf/pda/Mattern88} for
further discussion and physical interpretations.

To match this decompositional view with algebraic operations, the
source interface of an iposet is formed by some of its minimal events
(unless it is empty) and the target interface by some of its maximal
ones. Our gluing composition acts like standard serial composition
outside of interfaces, but gluing interface events together. The
particular matching of interface events is determined by numbering
them, and this and the resulting renumbering makes parallel
composition non-commutative. The gluing composition has many units,
whereas the parallel composition has a single one---the empty poset.
Yet both operations are associative. The standard serial and parallel
compositions of posets are recovered when iposets have empty
interfaces.

Using these basic algebraic properties we show that iposets with
gluing and parallel compositions and a suitable generalisation of the
subsumption order form a strict 2-category with parallel composition
as a lax tensor. This yields a lax interchange law between gluing
and parallel composition.  The notion of strict 2-categories with lax
tensors is new, but closely related to standard 2-categorical
structures. It generalises the concurrent monoids that capture the
equational theory of series-parallel pomsets.

By analogy to series-parallel posets, we also define a hierarchy of
gluing-parallel iposets and show that it does not collapse. We
identify forbidden substructures in order to show that it is a proper
subclass of the class of all iposets. We relate it with the hierarchy
of series-parallel pomsets as well as with interval orders. All
series-parallel posets are in the gluing-parallel hierarchy. The
interval orders, which allow $\N$-shaped posets, are captured at its
first alternation level.  Interval orders arise naturally in geometric
realisation of higher-dimensional automata or in situations where
events in concurrent systems extend in time
\cite{conf/pda/Mattern88,DBLP:journals/acta/GlabbeekG01,
  Pratt03trans_cancel}, such as weak and transactional memory systems
\cite{DBLP:journals/toplas/HerlihyW90}.  A precise geometric
characterisation of gluing-parallel iposets in terms of forbidden
substructures remains open.  Further, the precise algebraic setting in
which they would be freely generated, analogous to series-parallel
posets being free concurrent monoids, is open.  (The $2$-category
outlined above does not suffice: certain iposets satisfy a strong
interchange law.)

This article is based on a previous conference paper
\cite{DBLP:conf/RelMiCS/FahrenbergJST20}.  The description of iposets
in terms of strict 2-categories in Definition \ref{de:laxtensor} and
Theorem \ref{T:ipos2cat} is one of them.  A sufficient condition for
the existence of gluing decompositions in Lemma
\ref{le:criterion-for-gluing} is another one.  We also make the
hierarchies in Section \ref{se:hierarchy} more precise and add several
new forbidden substructures in Proposition \ref{pr:forbidden2}. An
erroneous claim about freeness in a particular algebraic class
\cite[Thm.~19]{DBLP:conf/RelMiCS/FahrenbergJST20} is refuted in
Example \ref{ex:antilevi}.

\section{Posets}
\label{se:posets}
\label{S:interval-orders}

We assume basic knowledge of order theory, see \cite{DaveyP12} for
details.  We restrict our attention to finite posets with strict
orderings, as is usual in the partial order semantics of
concurrency. Throughout this paper, a poset $(P,\mathord<)$ is thus a
finite set $P$ equipped with an irreflexive transitive binary relation
$<$ (asymmetry of $<$ follows).  In Hasse diagrams, we put greater
elements to the right of smaller ones.  In the opposite poset
$P^\op=(P,>)$ of a poset $(P,<)$, the order is reversed.

We write $[n]=\{1, \dotsc, n\}$, for $n\ge 1$, both for the set with
elements $1$ to $n$ and for the discrete poset on these points where
all points are incomparable. Additionally, $[0]=\emptyset$.  We write
$[\po n]=\{1<\dotsm<n\}$ for the linear order on $[n]$. In particular,
$[\po 1]=[1]$.

A function $f:P\to Q$ between posets $(P,{<}_P)$ and $(Q,{<}_Q)$ is
order-preserving if $f( x)<_Q f( y)$ whenever $x<_P y$, and
order-reflecting if $x<_P y$ whenever $f( x)<_Q f( y)$.  A poset
monomorphism is an order-preserving and order-reflecting injection. A
poset $P$ is an induced subposet of $Q$ if there exists a monomorphism
$P\hookrightarrow Q$.  Poset isomorphisms are bijective monomorphisms;
we write $P\cong Q$ if $P$ and $Q$ are isomorphic.

We usually consider posets, and operations on them, up to isomorphism.
Intuitively this means that we are not interested in the identity of
the events in a poset $(P, <)$ but rather in their order structure
given by precedence relation $<$.  The set of isomorphism classes of
posets is denoted by $\Pos$.

\begin{remark}
  By contrast to $\le$-embeddings,
  $<$-preserving and reflecting functions need not be injective: the only
  $<$-preserving function from the V-shape three-point poset to
  $[\po 2]$ is $<$-reflecting (and surjective), but has no
  inverse. Every $<$-preserving function is nevertheless injective on
  comparable points.
\end{remark}

In concurrency theory, posets are often equipped with a serial and a
parallel composition \cite{DBLP:journals/ipl/Winkowski77}. Both are
based on the disjoint union (coproduct) of sets, defined as
$X\sqcup Y = \{(x,1)\mid x\in X\}\cup \{(y,2)\mid y \in Y\}$.

\begin{definition}
  Let $(P_1,{<}_1)$ and $(P_2,{<}_2)$ be posets.
  \begin{enumerate}
  \item Their \emph{parallel composition} $P_1\otimes P_2$ is the
    coproduct with $P_1\sqcup P_2$ as carrier set and order defined
    as
    \begin{equation*}
      (p,i) < (q,j) \ \Leftrightarrow\  i=j \land p<_i q,\qquad i,j\in
      \{1,2\}.
    \end{equation*}
  \item Their \emph{serial composition} $P_1\pomser P_2$ is the
    ordinal sum, which again has the disjoint union as carrier set,
    but order defined as
    \begin{equation*}
      (p,i) < (q,j) \ \Leftrightarrow \ (i=j \land p<_i q) \lor i<j,
      \qquad i,j\in \{1,2\}.
    \end{equation*}
  \end{enumerate}
\end{definition}

Intuitively, $P_1\otimes P_2$ puts the Hasse diagram of $P_1$ above
that of $P_2$, whereas $P_1\pomser P_2$ puts the Hasse diagram of
$P_1$ to the left of that of $P_2$ and adds arrows from each element
of $P_1$ to each element of $P_2$.  It is clear that both operations
are well defined on isomorphism classe. They are associative and have
the empty poset as their unit, up to isomorphism. Parallel composition
is commutative while serial composition is not.  Isomorphism classes
of posets thus form a monoid with respect to serial composition and a
commutative monoid with respect to concurrent composition. These
monoids share their unit (the empty pomset).

\begin{figure}[tbp]
  \centering
  \begin{tikzpicture}[x=1.3cm, y=.8cm] 
    \begin{scope}
      \begin{scope}
        \path[use as bounding box] (0,0) to (1,-1);
        \node (1) at (0,0) {\intpt};
        \node (2) at (1,0) {\intpt};
        \node (3) at (0,-1) {\intpt};
        \node (4) at (1,-1) {\intpt};
        \path (1) edge (2);
        \path (3) edge (2);
        \path (3) edge (4);
        \node at (1.5,-.5) {$\otimes$};
      \end{scope}
      \begin{scope}[shift={(2,0)}]
        \path[use as bounding box] (0,0) to (1,-1);
        \node (1) at (0,0) {\intpt};
        \node (2) at (1,0) {\intpt};
        \node (3) at (0,-1) {\intpt};
        \node (4) at (1,-1) {\intpt};
        \path (1) edge (2);
        \path (3) edge (2);
        \path (3) edge (4);
        \node at (1.5,-.5) {$=$};
      \end{scope}
      \begin{scope}[shift={(4.5,-.1)}, yscale=.8]
        \begin{scope}[shift={(0,-1)}]
          \path[use as bounding box] (0,0) to (1,-1);
          \node (1) at (0,0) {\intpt};
          \node (2) at (1,0) {\intpt};
          \node (3) at (0,-1) {\intpt};
          \node (4) at (1,-1) {\intpt};
          \path (1) edge (2);
          \path (3) edge (2);
          \path (3) edge (4);
        \end{scope}
        \begin{scope}[shift={(0,1)}]
          \path[use as bounding box] (0,0) to (1,-1);
          \node (1) at (0,0) {\intpt};
          \node (2) at (1,0) {\intpt};
          \node (3) at (0,-1) {\intpt};
          \node (4) at (1,-1) {\intpt};
          \path (1) edge (2);
          \path (3) edge (2);
          \path (3) edge (4);
        \end{scope}
      \end{scope}
    \end{scope}
    \begin{scope}[shift={(0,-3)}]
      \begin{scope}
        \path[use as bounding box] (0,0) to (1,-1);
        \node (1) at (0,0) {\intpt};
        \node (2) at (1,0) {\intpt};
        \node (3) at (0,-1) {\intpt};
        \node (4) at (1,-1) {\intpt};
        \path (1) edge (2);
        \path (3) edge (2);
        \path (3) edge (4);
        \node at (1.5,-.5) {$\pomser$};
      \end{scope}
      \begin{scope}[shift={(2,0)}]
        \path[use as bounding box] (0,0) to (1,-1);
        \node (1) at (0,0) {\intpt};
        \node (2) at (1,0) {\intpt};
        \node (3) at (0,-1) {\intpt};
        \node (4) at (1,-1) {\intpt};
        \path (1) edge (2);
        \path (3) edge (2);
        \path (3) edge (4);
        \node at (1.5,-.5) {$=$};
      \end{scope}
      \begin{scope}[shift={(4,0)}]
        \node (1) at (0,0) {\intpt};
        \node (2) at (1,0) {\intpt};
        \node (3) at (2,0) {\intpt};
        \node (4) at (3,0) {\intpt};
        \node (5) at (0,-1) {\intpt};
        \node (6) at (1,-1) {\intpt};
        \node (7) at (2,-1) {\intpt};
        \node (8) at (3,-1) {\intpt};
        \foreach \i/\j in {1/2,2/3,2/7,3/4,5/2,5/6,6/3,6/7,7/4,7/8}
        \path (\i) edge (\j);
      \end{scope}
    \end{scope}
  \end{tikzpicture}
  \caption{Parallel and serial compositions of posets $\N$.}
  \label{fi:exparser}
\end{figure}
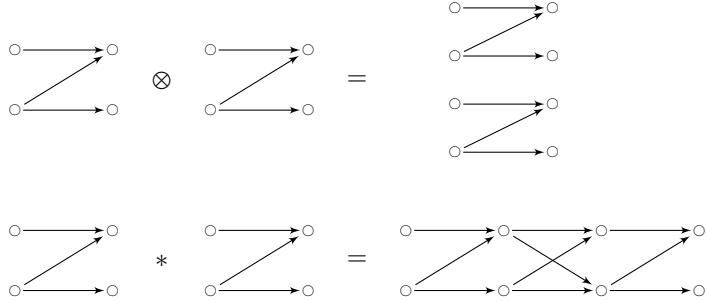

\begin{example}
  For any $n,m\ge 0$, it holds that $[n]\otimes[m]\cong[n+m]$ and
  $[\poo n]\pomser[\poo m]\cong[\poo{n+m}]$; thus
  $[n]\cong[1]\otimes\dotsm\otimes[1]$ and
  $[\poo n]\cong[1]\pomser\dotsm\pomser[1]$ ($n$ times each).  Figure
  \ref{fi:exparser} shows further examples.
\end{example}

A poset is \emph{series-parallel} (an \emph{sp-poset}) if it is either
empty or generated from the singleton poset by finitely many serial
and parallel compositions.  As mentioned in the introduction,
sp-posets are precisely those posets that do not contain $\N$ as an
induced subposet.

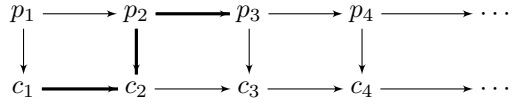
\begin{figure}[tbp]
  \centering
  \begin{tikzpicture}[x=1.5cm, shorten <=-1pt, shorten >=-1pt]
    \foreach \x in {1,2,3,4} \node (p\x) at (\x,0) {$p_\x$};
    \foreach \x in {1,2,3,4} \node (c\x) at (\x,-1) {$c_\x$};
    \node (p5) at (5.2,0) {$\cdots$};
    \node (c5) at (5.2,-1) {$\cdots$};
    \path[-, very thick] (c1) edge (c2) (p2) edge (c2) (p2) edge (p3);
    \foreach \x in {1,2,3,4} \path (p\x) edge (c\x);
    \foreach \i/\j in {1/2,2/3,3/4,4/5} \path (p\i) edge (p\j);
    \foreach \i/\j in {1/2,2/3,3/4,4/5} \path (c\i) edge (c\j);
  \end{tikzpicture}
  \caption{Producer-consumer pomset with an induced subposet $\N$
    indicated by bold arrows.}
  \label{fi:prodcon}
\end{figure}
This makes sp-posets unsuitable for many applications: even simple
producer-consumer systems generate $\N$'s
\cite{DBLP:journals/tcs/LodayaW00} and their structure cannot be
captured by sp-posets, see Figure~\ref{fi:prodcon}.

Interval orders \cite{Wiener14, journals/mpsy/Fishburn70} form another
class of posets relevant to concurrent and distributed computing
\cite{DBLP:conf/apn/Janicki18}. Intuitively, they are isomorphic to
sets of intervals on the real line that are ordered whenever they do
not overlap.  Interval orders can therefore capture events that extend
in time.

\begin{definition}
  An \emph{interval order} is a poset $(P,<)$ such that $w< y$ and
  $x< z$ imply $w< z$ or $x< y$, for all $w,x,y,z\in P$.
\end{definition}

(Transitivity of $<$ follows from asymmetry and the property above.)
Geometrically, there is once again a single forbidden substructure:
interval orders are precisely those posets that do not contain an
induced subposet of the form
\begin{equation*}
  \twotwo
  = [\poo{2}]\otimes[\poo{2}]
  =\! \vcenter{\hbox{%
      \begin{tikzpicture}[y=.5cm]
        \node (0) at (0,0) {\intpt};
        \node (1) at (0,-1) {\intpt};
        \node (2) at (1,0) {\intpt};
        \node (3) at (1,-1) {\intpt};
        \path (0) edge (2) (1) edge (3);
      \end{tikzpicture}}}.
\end{equation*}

The intuition outlined above is captured by Fishburn's representation
theorem \cite{journals/mpsy/Fishburn70, book/Fishburn85}. A poset $P$
is an interval order if and only if it has an \emph{interval
  representation}: a map $I:P\to 2^\Real$ assigning to each $x\in P$ a
closed real interval $I(x)=[b(x),e(x)]$, with $b(x)\le e(x)$, such
that $y<_P z$ if and only if $e(y)<_\Real b(z)$, for all $y,z\in P$.

Each interval order admits an interval representation with a minimal
number of endpoints. Such minimal representations are in bijective
correspondence with the \emph{closed interval traces} of concurrency
theory~\cite{DBLP:journals/iandc/JanickiY17}, as shown
in~\cite{DBLP:conf/RelMiCS/FahrenbergJST20}. Closed interval sequences
are finite sequences of $b(x)$ and $e(x)$ with $x$ ranging over some
finite set, where each $b(x)$ and $e(x)$ occurs exactly once and each
$e(x)$ after the corresponding $b(x)$. Closed interval traces are
equivalence classes of such sequences modulo the relations
$b(x)b(y)\approx b(y)b(x)$ and $e(x)e(y)\approx e(y)e(x)$.  They are
order-theoretic analogues of \emph{ST-traces} of Petri nets
\cite{DBLP:conf/ifip2/Glabbeek90, DBLP:books/sp/Vogler92}.

We show in Section~\ref{se:hierarchy} that interval orders appear at
the first level of the alternation hierarchy of the gluing-parallel
posets with interfaces introduced in the following sections.

\section{Posets with Interfaces}
\label{se:iposets}

We now define posets with interfaces, their gluing and parallel
composition and the units of these operations. We show that posets
with interfaces, up to isomorphism, form a category with respect to
gluing composition and that the presence of interfaces makes our
parallel composition non-commutative.

\begin{definition}
  A \emph{poset with interfaces} (\emph{iposet}) is a poset $(P,<)$
  equipped with two injective morphisms
  \begin{equation*}
    [n]\overset{s}{\longrightarrow}
    P\overset{t}{\longleftarrow}[m],\qquad n,m\ge 0,
  \end{equation*}
  such that the elements in the image $s([n])$ are minimal and those
  in $t([m])$ maximal in $P$.
\end{definition}

The injections $[n]\overset{s}{\to} P \overset{t}{\leftarrow} [m]$ represent the
\emph{source} and the \emph{target interface} of $P$, respectively.
We write $S_P$ instead of $s_P([n])$ and $T_P$ instead of $t_P([m])$
for sets of interface elements and drop indices if convenient. We also
write $(s,P,t):n\to m$ instead of
$[ n]\overset{s}{\longrightarrow} P\overset{t}{\longleftarrow}[ m]$;
even $P:n\to m$ when the interfaces are clear.  The \emph{opposite} of
an iposet $(s,P,t):n\to m$ is the iposet $(t,P^\op,s):m\to n$, also
denoted by $P^\op$, where the order has been reversed and the source
and target interfaces have been swapped.

Figure \ref{fi:iposets} shows examples of iposets. Interface elements
are represented as half-circles to indicate the incomplete nature of
the corresponding events.

\begin{definition}
  A \emph{subsumption} of iposets $(s_1,P_1,t_1),(s_2,P_2,t_2):n\to m$
  is an order-reflecting bijection $f:P_1\to P_2$ that preserves
  interfaces: $f(x)<_2 f(y)$ implies $x<_1 y$ for all $x,y\in P_1$,
  $f\circ s_1= s_2$, and $f\circ t_1= t_2$.
  \begin{equation*}
    \begin{tikzcd}[row sep=tiny]
      & P_1\arrow[dd,"f"] &\\
      \left[n\right]\arrow[ur, "s_1"]\arrow[dr, "s_2" '] &&
      \left[m\right]\arrow[ul, "t_1" '] \arrow[dl, "t_2"] \\
      & P_2 &
    \end{tikzcd}
  \end{equation*}
  An iposet \emph{isomorphism} is an order-preserving subsumption.
\end{definition}

We write $P_1\subsu P_2$ if there exists a subsumption
$P_1\to P_2$ and $P_1\cong P_2$ if there exists an iposet
isomorphism $P_1\to P_2$. The relation $\subsu$ is a preorder
and $\mathord{\cong} = \mathord{\preceq} \cap \mathord{\succeq}$
(proving this requires some thought~\cite{journals/tcs/Esik02}).

\begin{remark}
  The subsumption relation $\subsu$ for posets is well studied
  \cite{DBLP:journals/fuin/Grabowski81, DBLP:journals/tcs/Gischer88}.
  In the literature, $P_1\subsu P_2$ is usually defined by existence
  of an order-preserving bijection $P_2\to P_1$. This is equivalent to
  our definition.
\end{remark}

Like for posets above, we usually consider iposets up to isomorphism.
The subsumption relation $\preceq$ is a partial order on isomorphism
classes. Intuitively, $P_1\preceq P_2$ holds if $P_1$ and $P_2$ have
the same points and interfaces and $P_1$ is at least as ordered as
$P_2$.

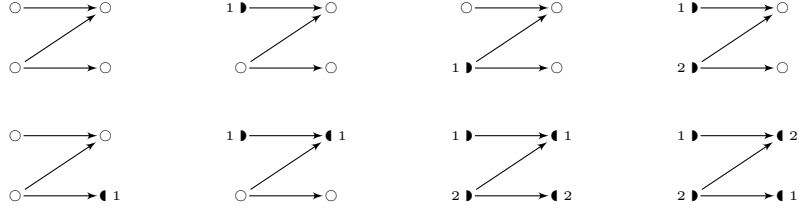
\begin{figure}[tbp]
  \centering
  \begin{tikzpicture}[x=1.2cm, y=.8cm]
    \begin{scope}
      \begin{scope}
        \path[use as bounding box] (0,0) to (1,-1);
        \node (1) at (0,0) {\intpt};
        \node (2) at (1,0) {\intpt};
        \node (3) at (0,-1) {\intpt};
        \node (4) at (1,-1) {\intpt};
        \path (1) edge (2);
        \path (3) edge (2);
        \path (3) edge (4);
      \end{scope}
      \begin{scope}[xshift=3cm]
        \path[use as bounding box] (0,0) to (1,-1);
        \node [label=left:{\tiny $1$}] (1) at (0,0) {\inpt};
        \node (2) at (1,0) {\intpt};
        \node (3) at (0,-1) {\intpt};
        \node (4) at (1,-1) {\intpt};
        \path (1) edge (2);
        \path (3) edge (2);
        \path (3) edge (4);
      \end{scope}
      \begin{scope}[xshift=6cm]
        \path[use as bounding box] (0,0) to (1,-1);
        \node (1) at (0,0) {\intpt};
        \node (2) at (1,0) {\intpt};
        \node [label=left:{\tiny $1$}] (3) at (0,-1) {\inpt};
        \node (4) at (1,-1) {\intpt};
        \path (1) edge (2);
        \path (3) edge (2);
        \path (3) edge (4);
      \end{scope}
      \begin{scope}[xshift=9cm]
        \path[use as bounding box] (0,0) to (1,-1);
        \node [label=left:{\tiny $1$}] (1) at (0,0) {\inpt};
        \node (2) at (1,0) {\intpt};
        \node [label=left:{\tiny $2$}] (3) at (0,-1) {\inpt};
        \node (4) at (1,-1) {\intpt};
        \path (1) edge (2);
        \path (3) edge (2);
        \path (3) edge (4);
      \end{scope}
    \end{scope}
    \begin{scope}[yshift=-1.7cm]
      \begin{scope}
        \path[use as bounding box] (0,0) to (1,-1);
        \node (1) at (0,0) {\intpt};
        \node (2) at (1,0) {\intpt};
        \node (3) at (0,-1) {\intpt};
        \node [label=right:{\tiny $1$}] (4) at (1,-1) {\outpt};
        \path (1) edge (2);
        \path (3) edge (2);
        \path (3) edge (4);
      \end{scope}
      \begin{scope}[xshift=3cm]
        \path[use as bounding box] (0,0) to (1,-1);
        \node [label=left:{\tiny $1$}] (1) at (0,0) {\inpt};
        \node [label=right:{\tiny $1$}] (2) at (1,0) {\outpt};
        \node (3) at (0,-1) {\intpt};
        \node (4) at (1,-1) {\intpt};
        \path (1) edge (2);
        \path (3) edge (2);
        \path (3) edge (4);
      \end{scope}
      \begin{scope}[xshift=6cm]
        \path[use as bounding box] (0,0) to (1,-1);
        \node [label=left:{\tiny $1$}]  (1) at (0,0) {\inpt};
        \node [label=right:{\tiny $1$}] (2) at (1,0) {\outpt};
        \node [label=left:{\tiny $2$}] (3) at (0,-1) {\inpt};
        \node [label=right:{\tiny $2$}] (4) at (1,-1) {\outpt};
        \path (1) edge (2);
        \path (3) edge (2);
        \path (3) edge (4);
      \end{scope}
      \begin{scope}[xshift=9cm]
        \path[use as bounding box] (0,0) to (1,-1);
        \node [label=left:{\tiny $1$}]  (1) at (0,0) {\inpt};
        \node [label=right:{\tiny $2$}] (2) at (1,0) {\outpt};
        \node [label=left:{\tiny $2$}] (3) at (0,-1) {\inpt};
        \node [label=right:{\tiny $1$}] (4) at (1,-1) {\outpt};
        \path (1) edge (2);
        \path (3) edge (2);
        \path (3) edge (4);
      \end{scope}
    \end{scope}
  \end{tikzpicture}
  \caption{Eight of 25 different iposets based on poset $\N$.}
  \label{fi:iposets}
\end{figure}

Next we introduce a partial sequential gluing composition on iposets,
defined whenever the interfaces on ends agree.  We also adapt the
standard parallel composition of posets to iposets.  The latter
requires the isomorphisms $\phi_{ n, m}:[ n+ m]\to[ n]\otimes[ m]$
given by
\begin{equation*}
  \phi_{ n, m}( i)=
  \begin{cases}
    (i,1) &\text{if } i\le n, \\
    (i-n,2) &\text{if } i> n;
  \end{cases}
\end{equation*}
we also tacitly use parallel composition of poset morphisms, which is
defined in the obvious way.

\begin{definition}
  \label{de:parglue}
  Let $(s_1,P_1,t_1):n_1\to m_1$ and $(s_2,P_2 ,t_2):n_2\to m_2$ be
  iposets.
  \begin{enumerate}
  \item Their \emph{parallel composition} is the iposet
    $P_1\otimes P_2 = (s,P_1\otimes P_2,t):n_1+n_2\to m_1+m_2$ with
    $s=(s_1\otimes s_2)\circ \phi_{ n_1, n_2}$ and
    $t=(t_1\otimes t_2)\circ \phi_{ m_1, m_2}$.
  \item For $m_1=n_2$, their \emph{gluing composition} is the iposet
    $P_1\glue P_2 = (s_1,P,t_2):n_1\to m_2$,
    where the carrier set is the quotient
    \begin{equation*}
     P =(P_1\sqcup P_2)/\{(t_1(k),1)= (s_2(k),2)\}_{k\in [m_1]}
    \end{equation*}
    and the order is defined, for $i,j\in \{1,2\}$, as
    \begin{equation*}
      (p,i)<(q,j) \Leftrightarrow (i=j \land p<_i q) \lor (i<j \land
      p\notin T_{P_1} \land q\notin S_{P_2}).
    \end{equation*}
  \end{enumerate}
\end{definition}

Intuitively, parallel composition of iposets puts the Hasse diagrams
of the underlying posets on top of each other while renumbering
interfaces.  Gluing composition $P_1\glue P_2$ puts the Hasse diagram
of $P_1$ to the left of that of $P_2$, whenever the target interface
of $P_1$ and the source interface of $P_2$ match. It then glues
corresponding interface points and adds arrows from all points in
$P_1$ that are not in its target interface to all points in $P_2$ that
are not in its source interface.  As explained in the introduction, it
thus glues events of $P_1$ that do not finish in $P_1$ with those
events of $P_2$ that do not start in $P_2$.

It is clear that $\otimes$ and $\glue$ are well defined on isomorphism
classes: they are associative up to isomorphism.  Figures
\ref{fi:ndecomp} and \ref{fi:iposetcomp} show examples of gluing
(de)compositions of iposets, including the $\N$. The half-circles in
source and target interfaces are glued to full circles in these
diagrams.

\begin{figure}[tbp]
  \centering
  \begin{tikzpicture}[x=1.3cm, y=.8cm]
   \begin{scope}
      \node (1) at (0,0) {\intpt};
      \node (2) at (1,0) {\intpt};
      \node (3) at (0,-1) {\intpt};
      \node (4) at (1,-1) {\intpt};
      \path (1) edge (2);
      \path (3) edge (2);
      \path (3) edge (4);
      \node at (1.5,-.5) {$=$};
    \end{scope}
    \begin{scope}[shift={(2,0)}]
      \node [label=right:{\tiny 1}] (1) at (0,0) {\outpt};
      \node (2) at (0,-1) {\intpt}; 
      \node at (.5,-.5) {$\glue$};
    \end{scope}
    \begin{scope}[shift={(3,0)}]
      \node [label=left:{\tiny 1}] (1) at (0,0) {\inpt};
      \node [label=right:{\tiny 1}] (2) at (0,-1) {\outpt}; 
      \node at (.5,-.5) {$\glue$};
    \end{scope}
    \begin{scope}[shift={(4,0)}]
      \node (1) at (0,0) {\intpt};
      \node [label=left:{\tiny 1}] (2) at (0,-1) {\inpt}; 
    \end{scope}
  \end{tikzpicture}
  \caption{Decomposition of poset $\N$.}
  \label{fi:ndecomp}
\end{figure}

\begin{figure}[tbp]
  \centering
  \begin{tikzpicture}[x=1.3cm, y=.8cm]
    \begin{scope}[shift={(0,-2)}]
      \begin{scope}
        \path[use as bounding box] (0,0) to (1,-1);
        \node (1) at (0,0) {\intpt};
        \node [label=right:{\tiny $1$}] (2) at (1,0) {\outpt};
        \node (3) at (0,-1) {\intpt};
        \node (4) at (1,-1) {\intpt};
        \path (1) edge (2);
        \path (3) edge (2);
        \path (3) edge (4);
        \node at (1.5,-.5) {$\glue$};
      \end{scope}
      \begin{scope}[shift={(2,0)}]
        \path[use as bounding box] (0,0) to (1,-1);
        \node [label=left:{\tiny $1$}] (1) at (0,0) {\inpt};
        \node (2) at (1,0) {\intpt};
        \node (3) at (0,-1) {\intpt};
        \node (4) at (1,-1) {\intpt};
        \path (1) edge (2);
        \path (3) edge (2);
        \path (3) edge (4);
        \node at (1.5,-.5) {$=$};
      \end{scope}
      \begin{scope}[shift={(4,0)}]
        \node (1) at (0,0) {\intpt};
        \node (23) at (1.5,0) {\intpt};
        \node (4) at (3,0) {\intpt};
        \node (5) at (0,-1) {\intpt};
        \node (6) at (1,-1) {\intpt};
        \node (7) at (2,-1) {\intpt};
        \node (8) at (3,-1) {\intpt};
        \foreach \i/\j in {1/23,1/7,23/4,5/23,5/6,6/7,7/4,7/8}
        \path (\i) edge (\j);
      \end{scope}
    \end{scope}
    \begin{scope}[shift={(0,-4)}]
      \begin{scope}
        \path[use as bounding box] (0,0) to (1,-1);
        \node (1) at (0,0) {\intpt};
        \node [label=right:{\tiny $1$}] (2) at (1,0) {\outpt};
        \node (3) at (0,-1) {\intpt};
        \node (4) at (1,-1) {\intpt};
        \path (1) edge (2);
        \path (3) edge (2);
        \path (3) edge (4);
        \node at (1.5,-.5) {$\glue$};
      \end{scope}
      \begin{scope}[shift={(2,0)}]
        \path[use as bounding box] (0,0) to (1,-1);
        \node (1) at (0,0) {\intpt};
        \node (2) at (1,0) {\intpt};
        \node  [label=left:{\tiny $1$}] (3) at (0,-1) {\inpt};
        \node (4) at (1,-1) {\intpt};
        \path (1) edge (2);
        \path (3) edge (2);
        \path (3) edge (4);
        \node at (1.5,-.5) {$=$};
      \end{scope}
      \begin{scope}[shift={(4,0)}]
        \node (1) at (0,0) {\intpt};
        \node (23) at (1.5,0) {\intpt};
        \node (3) at (2,0) {};
        \node (4) at (3,0) {\intpt};
        \node (5) at (0,-1) {\intpt};
        \node (6) at (1,-1) {\intpt};
        \node (7) at (2,-1) {\intpt};
        \node (8) at (3,-1) {\intpt};
        \foreach \i/\j in {1/23,1/7,23/4,5/23,5/6,6/4,6/7,23/8,7/8}
        \path (\i) edge (\j);
      \end{scope}
    \end{scope}
    \begin{scope}[shift={(0,-6)}]
      \begin{scope}
        \path[use as bounding box] (0,0) to (1,-1);
        \node (1) at (0,0) {\intpt};
        \node  [label=right:{\tiny $1$}] (2) at (1,0) {\outpt};
        \node (3) at (0,-1) {\intpt};
        \node  [label=right:{\tiny $2$}] (4) at (1,-1) {\outpt};
        \path (1) edge (2);
        \path (3) edge (2);
        \path (3) edge (4);
        \node at (1.5,-.5) {$\glue$};
      \end{scope}
      \begin{scope}[shift={(2,0)}]
        \path[use as bounding box] (0,0) to (1,-1);
        \node  [label=left:{\tiny $1$}]  (1) at (0,0) {\inpt};
        \node (2) at (1,0) {\intpt};
        \node  [label=left:{\tiny $2$}] (3) at (0,-1) {\inpt};
        \node (4) at (1,-1) {\intpt};
        \path (1) edge (2);
        \path (3) edge (2);
        \path (3) edge (4);
        \node at (1.5,-.5) {$=$};
      \end{scope}
      \begin{scope}[shift={(4,0)}]
        \node (1) at (0,0) {\intpt};
        \node (23) at (1.5,0) {\intpt};
        \node (4) at (3,0) {\intpt};
        \node (5) at (0,-1) {\intpt};
        \node (67) at (1.5,-1) {\intpt};
        \node (8) at (3,-1) {\intpt};
        \foreach \i/\j in {1/23,23/4,5/67,5/23,67/8,67/4,1/8}
        \path (\i) edge (\j);
      \end{scope}
    \end{scope}
    \begin{scope}[shift={(0,-8)}]
      \begin{scope}
        \path[use as bounding box] (0,0) to (1,-1);
        \node (1) at (0,0) {\intpt};
        \node  [label=right:{\tiny $1$}] (2) at (1,0) {\outpt};
        \node (3) at (0,-1) {\intpt};
        \node  [label=right:{\tiny $2$}] (4) at (1,-1) {\outpt};
        \path (1) edge (2);
        \path (3) edge (2);
        \path (3) edge (4);
        \node at (1.5,-.5) {$\glue$};
      \end{scope}
      \begin{scope}[shift={(2,0)}]
        \path[use as bounding box] (0,0) to (1,-1);
        \node  [label=left:{\tiny $2$}]  (1) at (0,0) {\inpt};
        \node (2) at (1,0) {\intpt};
        \node  [label=left:{\tiny $1$}] (3) at (0,-1) {\inpt};
        \node (4) at (1,-1) {\intpt};
        \path (1) edge (2);
        \path (3) edge (2);
        \path (3) edge (4);
        \node at (1.5,-.5) {$=$};
      \end{scope}
      \begin{scope}[shift={(4,0)}]
        \node (1) at (0,0) {\intpt};
        \node (23) at (1.5,0) {\intpt};
        \node (4) at (3,0) {\intpt};
        \node (5) at (0,-1) {\intpt};
        \node (67) at (1.5,-1) {\intpt};
        \node (8) at (3,-1) {\intpt};
        \foreach \i/\j in {1/23,23/4,5/67,5/23,67/8,23/8}
        \path (\i) edge (\j);
      \end{scope}
    \end{scope}
  \end{tikzpicture}
  \caption{Gluings of different $\N$s with interfaces.}
  \label{fi:iposetcomp}
\end{figure}

\begin{proposition}
  \label{pr:ipos-comp}
  Iposets form a category $\iPos$ with natural numbers as objects,
  isomorphism classes of iposets $( s, P, t)\!: n\to m$ as morphisms,
  the \emph{identity iposets}
  $\id_n=( \id_{[n]}, [n], \id_{[n]})\!: n\to n$ as identity
  morphisms, $\glue$ as composition and $\dom,\cod:\iPos\to \Nat$
  sending $P:n\to m$ to $n$ and $m$, respectively, as domain and
  codomain maps.
\end{proposition}

\begin{proof}
  Checking the standard unit, domain and codomain axioms for
  categories is trivial; associativity of gluing composition (up to
  isomorphism) follows directly from unfolding Definition
  \ref{de:parglue}(2). \qed
\end{proof}

Posets may be regarded as iposets with empty interfaces. Thus, as
sets, $\Pos \cong \iPos(0,0)$, where $\iPos(0,0)$ denotes the $\iPos$
morphisms at object $0$. We write gluing composition in diagrammatical
order.

 serial poset composition for posets, $\glue$ is not commutative:
\begin{equation*}
  \routpt\! \rtone \;\glue\; \rtone \! \rinpt
  \ =\  \rintpt \ \ne
  \vcenter{\hbox{%
      \begin{tikzpicture}
        \node [label=left:{\tiny $1$}] (1) at (0,0) {\inpt};
        \node [label=right:{\tiny $1$}] (2) at (1,0) {\outpt};
        \path (1) edge (2);
      \end{tikzpicture}}}
  = \ \rtone\! \rinpt
  \;\glue\; \routpt\! \rtone
\end{equation*}
Parallel composition has the empty iposet as its unit.  As the next
example shows, it need not commute.

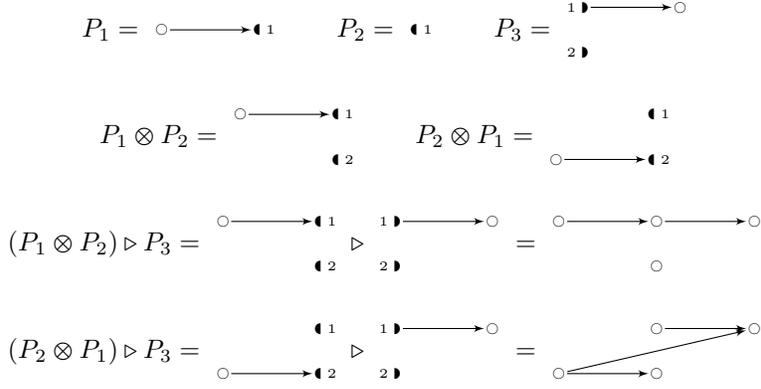
\begin{figure}
  \begin{gather*}
    P_1 = \vcenter{\hbox{%
        \begin{tikzpicture}[x=1.3cm, y=.6cm]
          \node (1) at (0,0) {\intpt};
          \node [label=right:{\tiny $1$}] (2) at (1,0) {\outpt};
          \path (1) edge (2);
        \end{tikzpicture}}}
    \qquad%
    P_2 = \vcenter{\hbox{%
        \begin{tikzpicture}[x=1.3cm, y=.6cm]
          \node [label=right:{\tiny $1$}] (2) at (0,0) {\outpt};
        \end{tikzpicture}}}
    \qquad%
    P_3 = \vcenter{\hbox{%
        \begin{tikzpicture}[x=1.3cm, y=.6cm]
          \node [label=left:{\tiny $1$}] (1) at (0,0) {\inpt};
          \node (2) at (1,0) {\intpt};
          \node [label=left:{\tiny $2$}] (3) at (0,-1) {\inpt};
          \path (1) edge (2);
        \end{tikzpicture}}}
    \\[2ex]
    P_1\otimes P_2 = \vcenter{\hbox{%
        \begin{tikzpicture}[x=1.3cm, y=.6cm]
          \node (1) at (0,0) {\intpt};
          \node [label=right:{\tiny $1$}] (2) at (1,0) {\outpt};
          \node [label=right:{\tiny $2$}] (3) at (1,-1) {\outpt};
          \path (1) edge (2);
        \end{tikzpicture}}}
    \qquad%
    P_2\otimes P_1 = \vcenter{\hbox{%
        \begin{tikzpicture}[x=1.3cm, y=.6cm]
          \node [label=right:{\tiny $1$}] (3) at (1,0) {\outpt};
          \node (1) at (0,-1) {\intpt};
          \node [label=right:{\tiny $2$}] (2) at (1,-1) {\outpt};
          \path (1) edge (2);
        \end{tikzpicture}}}
    \\[2ex]
    (P_1\otimes P_2)\glue P_3 = \vcenter{\hbox{%
        \begin{tikzpicture}[x=1.3cm, y=.6cm]
          \node (1) at (0,0) {\intpt};
          \node [label=right:{\tiny $1$}] (2) at (1,0) {\outpt};
          \node [label=right:{\tiny $2$}] (3) at (1,-1) {\outpt};
          \path (1) edge (2);
        \end{tikzpicture}}}
    \glue \vcenter{\hbox{%
        \begin{tikzpicture}[x=1.3cm, y=.6cm]
          \node [label=left:{\tiny $1$}] (1) at (0,0) {\inpt};
          \node (2) at (1,0) {\intpt};
          \node [label=left:{\tiny $2$}] (3) at (0,-1) {\inpt};
          \path (1) edge (2);
        \end{tikzpicture}}}
    = \vcenter{\hbox{%
        \begin{tikzpicture}[x=1.3cm, y=.6cm]
          \node (1) at (0,0) {\intpt};
          \node (2) at (1,0) {\intpt};
          \node (3) at (2,0) {\intpt};
          \node (4) at (1,-1) {\intpt};
          \path (1) edge (2);
          \path (2) edge (3);
        \end{tikzpicture}}}
    \\[2ex]
    (P_2\otimes P_1)\glue P_3 = \vcenter{\hbox{%
        \begin{tikzpicture}[x=1.3cm, y=.6cm]
          \node [label=right:{\tiny $1$}] (3) at (1,0) {\outpt};
          \node (1) at (0,-1) {\intpt};
          \node [label=right:{\tiny $2$}] (2) at (1,-1) {\outpt};
          \path (1) edge (2);
        \end{tikzpicture}}}
    \glue \vcenter{\hbox{%
        \begin{tikzpicture}[x=1.3cm, y=.6cm]
          \node [label=left:{\tiny $1$}] (1) at (0,0) {\inpt};
          \node (2) at (1,0) {\intpt};
          \node [label=left:{\tiny $2$}] (3) at (0,-1) {\inpt};
          \path (1) edge (2);
        \end{tikzpicture}}}
    = \vcenter{\hbox{%
        \begin{tikzpicture}[x=1.3cm, y=.6cm]
          \node (1) at (1,0) {\intpt};
          \node (2) at (2,0) {\intpt};
          \node (3) at (0,-1) {\intpt};
          \node (4) at (1,-1) {\intpt};
          \path (1) edge (2);
          \path (3) edge (4);
          \path (3) edge (2);
        \end{tikzpicture}}}
  \end{gather*}
  \caption{Non-isomorphic gluings of symmetric parallel compositions
    in Example \ref{ex:otimesnotsymm}.}
  \label{fi:otimesnotsymm}
\end{figure}

\begin{example}
  \label{ex:otimesnotsymm}
  Figure \ref{fi:otimesnotsymm} shows iposets $P_1$ and $P_2$ and
  their parallel products $P_1\otimes P_2$ and $P_2\otimes P_1$.  The
  latter are non-isomorphic as iposets because of the different
  labelling of the interfaces. Yet their underlying posets are
  isomorphic.  However, when they are glued with iposet $P_3$
  in the figure, the resulting posets $(P_1\otimes P_2)\glue P_3$ and
  $(P_2\otimes P_1)\glue P_3$ are non-isomorphic \emph{as posets}.
  Non-commutativity of parallel composition is therefore not a
  technical artefact of our definitions but inherent to the formalism.
\end{example}

Interfaces can be renumbered using special iposets called
\emph{symmetries}, as explained in Section \ref{se:variants}.  Lemma
\ref{le:symmetries} in that section shows that parallel composition of
iposets is commutative up to such symmetries; but following Example
\ref{ex:otimesnotsymm} above, enforcing such commutativity would make
gluing composition lose its congruence property, which is undesirable.
Further, Lemma \ref{le:gp-cons} below shows that gluing-parallel
iposets, the class of iposets we are mainly interested in, satisfy an
interface consistency property that rules out all non-trivial
symmetries.  Proposition \ref{P:otimes-comm} and Lemma
\ref{le:parcom-nointf} below describe precisely those parallel
compositions that commute.

\section{Interchange}
\label{S:algebra}

This and the next sections study the interaction of gluing and
parallel composition.  For posets, serial and parallel composition
interact through the lax interchange law
\begin{equation*}
  (P\otimes  P') \glue (Q\otimes Q') \subsu (P\glue
  Q)\otimes (P'\glue Q').
\end{equation*}
It equips the set of isomorphism classes of posets with a concurrent
monoid structure \cite{DBLP:journals/jlp/HoareMSW11}. However, iposets
with gluing composition form a category rather than a monoid; the
interaction with parallel composition thus requires a (strict)
$2$-category.  Readers unfamiliar with $2$-categories can skip the
following Section \ref{se:intercat} which is not relevant to the rest
of the paper.

A natural question is whether $\otimes$ yields a monoidal structure on
$\iPos$.  Yet the answer is negative: gluing and parallel composition
need not satisfy the interchange law a tensor$\otimes$ would require:

\begin{align*}
  (\rintpt  \;\otimes\; \rintpt) \glue (\rintpt
  \;\otimes\; \rintpt)%
  =\! \vcenter{\hbox{%
      \begin{tikzpicture}[y=.5cm]
        \node (0) at (0,0) {\intpt};
        \node (1) at (0,-1) {\intpt};
        \node (2) at (1,0) {\intpt};
        \node (3) at (1,-1) {\intpt};
        \path (0) edge (2) (1) edge (3) (0) edge (3) (1) edge (2);
      \end{tikzpicture}}}
 \ne\! \vcenter{\hbox{%
      \begin{tikzpicture}[y=.5cm]
        \node (0) at (0,0) {\intpt};
        \node (1) at (0,-1) {\intpt};
        \node (2) at (1,0) {\intpt};
        \node (3) at (1,-1) {\intpt};
        \path (0) edge (2) (1) edge (3);
      \end{tikzpicture}}}%
     = (\rintpt \;\glue\; \rintpt) \otimes (\rintpt \;\glue\;
  \rintpt).
\end{align*}

\begin{remark}
This differs from gluing compositions where interfaces of iposets are
defined by \emph{all} minimal and maximal elements
\cite{DBLP:journals/ipl/Winkowski77}.  It also differs from previous
serial compositions with interfaces \cite{DBLP:conf/birthday/FioreC13,
  DBLP:journals/corr/Mimram15}, where interfaces disappear and no
additional order is introduced.  The first case gives rise to strict
monoidal categories with a partially defined tensor, the other to
plain strict monoidal categories and, more specifically, PROPs.
\end{remark}

The following proposition introduces a lax interchange law for
iposets.

\begin{proposition}[Lax interchange]
  \label{P:lax-interchange}
  Let $P$, $P'$, $Q$ and $Q'$ be iposets such that $\cod(P)=\dom(Q)$ and
  $\cod(P')=\dom(Q')$.  Then $\cod(P\otimes P')=\dom(Q\otimes Q')$ and
  \begin{equation}
    \label{eq:lax-interchange}
    (P\otimes  P') \glue (Q\otimes Q') \subsu (P\glue
    Q)\otimes (P'\glue Q').
  \end{equation}
\end{proposition}

\begin{proof}
  The first claim holds because
  \begin{equation*}
    \cod(P\otimes P')=\cod(P)\otimes \cod(P')=\dom(Q)\otimes
    \dom(Q')=\dom(Q\otimes Q').
  \end{equation*}
  Set $L = (P\otimes P') \glue (Q\otimes Q')$ and
  $R = (P\glue Q)\otimes (P'\glue Q')$.  For the underlying sets, the
  definition of $\otimes$ implies that
  \begin{align*}
    R &= ( P\sqcup Q)_{/ t_P(i)= s_Q(i)}\sqcup( P'\sqcup Q')_{/
      t_{P'}(j)= s_{Q'}(j)} \\
    &= ( P\sqcup Q\sqcup P'\sqcup Q')_{/ t_P(i)= s_Q(i), t_{P'}(j)=
      s_{Q'}(j)} \\
    &= ( P\sqcup P'\sqcup Q\sqcup Q')_{/ t_{P\otimes P'}(k)=
      s_{Q\otimes Q'}(k)} = L.
  \end{align*}
  Both posets thus have the same carrier set, and we may choose $f:L\to R$
  to be the identity.  It remains to show that $f$ reflects the order:
  each arrow in $R$ must be in $L$.

  Suppose $x<_{R} y$, that is, $x<_{P\glue Q} y$ or
  $x<_{P'\glue Q'} y$.  In the first case, if $x<_P y$ or $x<_Q y$,
  then $x<_{P\otimes P'} y$ or $x<_{Q\otimes Q'} y$ and therefore
  $x<_{L} y$; and if $x\in P\setminus T_P$ and $y\in Q\setminus S_Q$,
  then $x\in P\sqcup P'\setminus T_{ P\otimes P'}$ and
  $y\in Q\sqcup Q'\setminus S_{ Q\otimes Q'}$ and therefore
  $x<_{L} y$, too. The second case is symmetric, and $x<_{L} y$ holds.\qed
\end{proof}

\section{The $2$-Category of Iposets}
\label{se:intercat}

We now introduce the $2$-categorical generalisation of concurrent
monoids which properly characterises our algebraic setting.  Readers
unfamiliar with $2$-categories can skip this section.

First, recall that the subsumption morphism used in the proof of
Proposition \ref{P:lax-interchange} is the identity on underlying
sets; in anticipation of the $2$-category structure introduced below
we used 2-cell notation for subsumptions and isomorphisms of iposets,
$f:P_1\To P_2$ instead of $f:P_1\to P_2$, in this section only.

\begin{corollary}
  \label{co:laxint}
  In the setting of Proposition \ref{P:lax-interchange}, we have the
  subsumption
  \begin{equation*}
    \id:(P\otimes  P') \glue (Q\otimes Q') \To (P\glue
    Q)\otimes (P'\glue Q').
  \end{equation*}
\end{corollary}

Next we extend gluing composition to subsumptions.

\begin{lemma}
  \label{de:comp-subsu}
  Let $f: P\To P'$ and $g: Q\To Q'$ be subsumptions
  such that $\cod(P)=\dom(Q)$ and $\cod(P')=\dom(Q')$.  Then
  $h= f\glue g: P\glue Q\To P'\glue Q'$ defined by
  \begin{equation*}
    h( x)=
    \begin{cases}
      f( x) &\text{if } x\in P, \\
      g( x) &\text{if } x\in Q
    \end{cases}
  \end{equation*}
  is well defined and a subsumption.
\end{lemma}

\begin{proof}
  Well-definedness of $h$ follows from the fact that subsumptions
  preserve interfaces, and $h$ is trivially a subsumption. \qed
\end{proof}

We define a notion of lax tensor for (strict) $2$-categories as a
specialisation of the notion of lax functor between bicategories, see
for example \cite[Sect.~4.1]{book/JohnsonY21}.  Recall that
$2$-categories are formed by objects ($0$-cells), morphisms
($1$-cells) and $2$-cells. Their axioms generalise those of
categories.  We write ``\,$;$\,'', in diagrammatical order, for
composition in a general $2$-category and ``$\To$'' for 2-cells.

\begin{definition}
  \label{de:laxtensor}
  A \emph{lax tensor} on a strict $2$-category $\mcal C$ is an
  operation
  $\mathord{\otimes}: \mcal C_{0, 1}\times \mcal C_{0, 1}\to \mcal
  C_{0, 1}$ on the $(0, 1)$-restriction of $\mcal C$, together with an
  object $\iota\in \mcal C_0$ and $2$-cells
  \begin{equation*}
    I_{f_1,g_1,f_2,g_2}: (f_1\otimes f_2); (g_1\otimes g_2)\To
    (f_1;g_1)\otimes(f_2;g_2)
  \end{equation*}
  for all morphisms $f_1:a_1\to b_1$, $g_1:b_1\to c_1$,
  $f_2:a_2\to b_2$, $g_2:b_2\to c_2$, such that
  \begin{enumerate}
  \item $f\otimes g: a\otimes c\to b\otimes d$ for all morphisms
    $f: a\to b$ and $g: c\to d$;
  \item $\id_a\otimes \id_b = \id_{a\otimes b}$ for all objects $a$
    and $b$;
  \item \label{en:laxtensor.unitass} $\iota\otimes a = a\otimes \iota = a$ and
    $a\otimes (b\otimes c) = (a\otimes b)\otimes c$ for all objects
    $a$, $b$, $c$;
  \item $\id_\iota\otimes f = f\otimes \id_\iota = f$ and
    $f\otimes (g\otimes h) = (f\otimes g)\otimes h$ for all morphisms
    $f$, $g$, $h$;
  \item as a mapping from tuples $(f_1,g_1,f_2,g_2)$ of morphisms to
    $2$-cells, $I$ is natural, \ie the diagram
    \begin{equation*}
    \begin{tikzcd}[row sep=large, column sep=large]
      \mcal C(a_1, b_1)\times \mcal C(a_2, b_2)\times \mcal
          C(b_1, c_1)\times \mcal C(b_2, c_2) \arrow[d, "\textstyle \mathord{\otimes}
          \times \mathord{\otimes}" ']\arrow[r,
          "\textstyle \mathord{;}\times\mathord{;}"] & \mcal C(a_1, c_1)\times
          \mcal C(a_2, c_2)\arrow[d, "\textstyle \otimes" ]\\
      \mcal C(a_1\otimes a_2, b_1\otimes b_2)\times \mcal
          C(b_1\otimes b_2, c_1\otimes c_2) \arrow[r, "\textstyle ;"
          ']\arrow[Rightarrow, ur, "\textstyle I" '] & \mcal C(a_1\otimes a_2, c_1\otimes c_2)
        \end{tikzcd}
    \end{equation*}
    commutes;
  \item \label{en:laxtensor.laxass} and $I$ satisfies lax associativity:
    \begin{equation*}
         \begin{tikzcd}[row sep=normal, column sep=small]
       (f_1\otimes f_2); (g_1\otimes g_2); (h_1\otimes
       h_2)\arrow[Rightarrow, dd, "\textstyle I_{f_1,g_1,f_2,g_2};
          \id_{h_1\otimes h_2}" ']\arrow[Rightarrow, drr, "\textstyle \id_{f_1\otimes f_2};
          I_{g_1,h_1,g_2,h_2}"] && \\
      &&  (f_1\otimes f_2); ((g_1;h_1)\otimes
      (g_2;h_2))\arrow[Rightarrow, dd, "\textstyle I_{f_1,g_1;h_1,f_2,g_2;h_2}"]\\
       ((f_1;g_1)\otimes (f_2;g_2)); (h_1\otimes
       h_2)\arrow[Rightarrow, drr, "\textstyle I_{f_1;g_1,h_1,f_2;g_2,h_2}" ']&&\\
      && (f_1;g_1;h_1)\otimes (f_2;g_2;h_2)\quad
    \end{tikzcd}
    \end{equation*}
  \end{enumerate}
\end{definition}

A lax tensor on $\mcal C$ is thus precisely a lax and strictly unital
functor $\mcal C\times \mcal C\to \mcal C$.  The ``interchanger''
$2$-cells $I_{f_1,g_1,f_2,g_2}$ replace the identities in the strict
interchange law, but are not required to be invertible.

We also need the following generalisation of acyclic or loop-free
categories \cite{journals/tac/Haucourt06, book/BridsonH99}.

\begin{definition}
  A $2$-category $\mcal C$ is \emph{$2$-acyclic} if it holds for all
  $f, g\in \mcal C_1$ that $\mcal C_2(f, g)\ne \emptyset$ and $\mcal
  C_2(g, f)\ne \emptyset$ imply $f=g$ and $\mcal C_2(f, g)=\{\id_f\}$.
\end{definition}

\begin{theorem}
  \label{T:ipos2cat}
  With the data given below, iposets form the $2$-acyclic strict
  $2$-category $\iPos$ with $\otimes$ as a lax tensor:
  \begin{itemize}
  \item objects: natural numbers;
  \item morphisms $n\to m$: iposets $(s,P,t):n\to m$;
  \item composition: $\glue:(n\to m)\times(m\to k)\to(n\to k)$;
  \item identities: $\id_n: n\to n$;
  \item $2$-cells $(n\to m)\To(n\to m)$: subsumptions;
  \item vertical $2$-composition: function composition of subsumptions;
  \item vertical $2$-identities: $\id:(n\to m)\To(n\to m)$;
  \item horizontal $2$-composition: gluing composition of subsumptions;
  \item horizontal $2$-identities: $\id: \id_n\To \id_n$.
  \end{itemize}
\end{theorem}

\begin{proof}
  The only properties of $\iPos$ that remain to be shown are
  associativity of horizontal $2$-composition and the fact that
  all $\id: \id_n\To \id_n$ are horizontal $2$-identities. Both are
  trivial.  Also $2$-acyclicity is clear given that $P\subsu Q$ and
  $Q\subsu P$ imply $P\cong Q$.

  To show that $\otimes$ is a lax tensor, Corollary \ref{co:laxint}
  implies that the interchanger $2$-cells
  $I_{P,Q,P',Q'}: (P\otimes P') \glue (Q\otimes Q') \To (P\glue
  Q)\otimes (P'\glue Q')$ are the identities on the underlying sets.
  Naturality and lax associativity are then immediate. \qed
\end{proof}

The $2$-category $\iPos$ generalises the concurrent
monoids mentioned in the introduction. The following proposition makes
this relationship precise.

\begin{proposition}
  Let $\mcal C$ be a $2$-acyclic strict $2$-category with lax tensor
  such that $\mcal C_0$ contains precisely one object.  With order
  defined by existence of\/ $2$-cells, $\mcal C_1$ forms a concurrent
  monoid.
\end{proposition}

\begin{proof}
  The $\otimes$-unit $\iota\in \mcal C_0$ is the unique object, and
  $;$ and $\otimes$ are operations on
  $\mcal C_1=\mcal C_1(\iota, \iota)$.  As $\mcal C_{0, 1}$ is a
  category, $;$ is associative and has $\iota$ as unit; by item
  \ref{en:laxtensor.unitass} of Definition \ref{de:laxtensor}, the
  same is true for $\otimes$.  Let $\le$ be the preorder on
  $\mcal C_1$ given by $f\le g$ if $\mcal C_2(f, g)\ne \emptyset$,
  then $\le$ is a partial order as $\mcal C_2$ is $2$-acyclic, and
  $(f\otimes g); (h\otimes k)\le (f; h)\otimes (g; k)$ for all
  $f, g, h, k\in \mcal C_1$ because of item \ref{en:laxtensor.laxass}
  of Definition \ref{de:laxtensor}. \qed
\end{proof}

In the setting of Theorem \ref{T:ipos2cat}, iposet isomorphisms are
$2$-cells. They should rightly be called $2$-isomorphisms; but we will
keep the simpler terminology.

\section{Starters, Terminators and Symmetries}
\label{se:variants}

This section and the next prepare our treatment of iposets generated
from singletons in Section \ref{S:gp}.  Here we discuss properties of
discrete iposets, which are important for defining non-trivial
(de)compositions of iposets and for understanding symmetries.

An iposet is \emph{discrete} if its order is empty.
A discrete iposet $(s,P,t):n\to m$ is a \emph{starter} if $t:[m]\to P$
is bijective. By opposition, it is a \emph{terminator} if $s:[n]\to P$
is bijective.

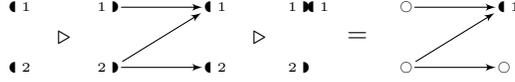
\begin{figure}[tbp]
  \centering
  \begin{tikzpicture}[y=.8cm, x=1.3cm]
    \begin{scope}
      \node [label=right:{\tiny $1$}] at (0,0) {\outpt};
      \node [label=right:{\tiny $2$}] at (0,-1) {\outpt};
      \node at (.5,-.5) {$\glue$};
    \end{scope}
    \begin{scope}[shift={(1,0)}]
      \node [label=left:{\tiny $1$}] (1) at (0,0) {\inpt};
      \node [label=right:{\tiny $1$}] (2) at (1,0) {\outpt};
      \node [label=left:{\tiny $2$}] (3) at (0,-1) {\inpt};
      \node [label=right:{\tiny $2$}] (4) at (1,-1) {\outpt};
      \path (1) edge (2);
      \path (3) edge (2);
      \path (3) edge (4);
      \node at (1.5,-.5) {$\glue$};
    \end{scope}
    \begin{scope}[shift={(3,0)}]
      \node [label=left:{\tiny $1$}] at (-.05,0) {\inpt};
      \node [label=left:{\tiny $2$}] at (-.05,-1) {\inpt};
      \node [label=right:{\tiny $1$}] at (.05,0) {\outpt};
      \node at (.5,-.5) {$=$};
    \end{scope}
    \begin{scope}[shift={(4,0)}]
      \node (1) at (0,0) {\intpt};
      \node [label=right:{\tiny $1$}] (2) at (1,0) {\outpt};
      \node (3) at (0,-1) {\intpt};
      \node (4) at (1,-1) {\intpt};
      \path (1) edge (2);
      \path (3) edge (2);
      \path (3) edge (4);
    \end{scope}
  \end{tikzpicture}
  \caption{Switching off interfaces by composition with a starter and
    a terminator.}
  \label{fi:switchoff}
\end{figure}

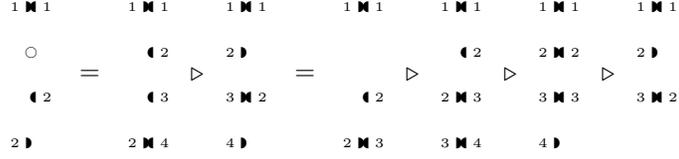
\begin{figure}[tbp]
  \centering
  \begin{tikzpicture}[y=.6cm, x=1.3cm]
    \begin{scope}
      \node [label=left:{\tiny $1$}] at (-.05,0) {\inpt};
      \node [label=right:{\tiny $1$}] at (.05,0) {\outpt};
      \node at (0,-1) {\intpt};
      \node [label=right:{\tiny $2$}] at (.05,-2) {\outpt};
      \node [label=left:{\tiny $2$}] at (-.05,-3) {\inpt};
      \node at (.6,-1.5) {$=$};
    \end{scope}
    \begin{scope}[shift={(1.2,0)}]
      \node [label=left:{\tiny $1$}] at (-.05,0) {\inpt};
      \node [label=right:{\tiny $1$}] at (.05,0) {\outpt};
      \node [label=right:{\tiny $2$}] at (.05,-1) {\outpt};
      \node [label=right:{\tiny $3$}] at (.05,-2) {\outpt};
      \node [label=left:{\tiny $2$}] at (-.05,-3) {\inpt};
      \node [label=right:{\tiny $4$}] at (.05,-3) {\outpt};
      \node at (.5,-1.5) {$\glue$};
    \end{scope}
    \begin{scope}[shift={(2.2,0)}]
      \node [label=left:{\tiny $1$}] at (-.05,0) {\inpt};
      \node [label=right:{\tiny $1$}] at (.05,0) {\outpt};
      \node [label=left:{\tiny $2$}] at (-.05,-1) {\inpt};
      \node [label=left:{\tiny $3$}] at (-.05,-2) {\inpt};
      \node [label=right:{\tiny $2$}] at (.05,-2) {\outpt};
      \node [label=left:{\tiny $4$}] at (-.05,-3) {\inpt};
      \node at (.6,-1.5) {$=$};
    \end{scope}
    \begin{scope}[shift={(3.4,0)}]
      \node [label=left:{\tiny $1$}] at (-.05,0) {\inpt};
      \node [label=right:{\tiny $1$}] at (.05,0) {\outpt};
      \node [label=right:{\tiny $2$}] at (.05,-2) {\outpt};
      \node [label=left:{\tiny $2$}] at (-.05,-3) {\inpt};
      \node [label=right:{\tiny $3$}] at (.05,-3) {\outpt};
      \node at (.5,-1.5) {$\glue$};
    \end{scope}
    \begin{scope}[shift={(4.4,0)}]
      \node [label=left:{\tiny $1$}] at (-.05,0) {\inpt};
      \node [label=right:{\tiny $1$}] at (.05,0) {\outpt};
      \node [label=right:{\tiny $2$}] at (.05,-1) {\outpt};
      \node [label=left:{\tiny $2$}] at (-.05,-2) {\inpt};
      \node [label=right:{\tiny $3$}] at (.05,-2) {\outpt};
      \node [label=left:{\tiny $3$}] at (-.05,-3) {\inpt};
      \node [label=right:{\tiny $4$}] at (.05,-3) {\outpt};
      \node at (.5,-1.5) {$\glue$};
    \end{scope}
    \begin{scope}[shift={(5.4,0)}]
      \node [label=left:{\tiny $1$}] at (-.05,0) {\inpt};
      \node [label=right:{\tiny $1$}] at (.05,0) {\outpt};
      \node [label=left:{\tiny $2$}] at (-.05,-1) {\inpt};
      \node [label=right:{\tiny $2$}] at (.05,-1) {\outpt};
      \node [label=left:{\tiny $3$}] at (-.05,-2) {\inpt};
      \node [label=right:{\tiny $3$}] at (.05,-2) {\outpt};
      \node [label=left:{\tiny $4$}] at (-.05,-3) {\inpt};
      \node at (.5,-1.5) {$\glue$};
    \end{scope}
    \begin{scope}[shift={(6.4,0)}]
      \node [label=left:{\tiny $1$}] at (-.05,0) {\inpt};
      \node [label=right:{\tiny $1$}] at (.05,0) {\outpt};
      \node [label=left:{\tiny $2$}] at (-.05,-1) {\inpt};
      \node [label=left:{\tiny $3$}] at (-.05,-2) {\inpt};
      \node [label=right:{\tiny $2$}] at (.05,-2) {\outpt};
    \end{scope}
  \end{tikzpicture}
  \caption{Decomposition of discrete iposet into starters and
    terminators.}
  \label{fi:discrete}
\end{figure}

Starters and terminators are useful for starting and terminating
individual events in gluing compositions: in $Q=S\glue P\glue T$ with
starter $S$ and terminator $T$, the poset of $Q$ equals that of $P$,
but parts of the interfaces of $P$ may have been switched off by $S$
and $T$.
As an example, Figure \ref{fi:switchoff} shows how a
starter is used to start the minimal events of an iposet $\N$ in which
they are already running, and a terminator is used to terminate one
unfinished event of the $\N$.

Every discrete iposet can obviously be obtained by gluing a starter
with a terminator, see also Figure \ref{fi:discrete}.

\begin{lemma}
  \label{le:trivial-deglue}
  Let $P:n\to m$ and $Q:m\to \ell$ be iposets.  Then $P\glue Q\cong Q$
  as posets if and only if $P$ is  a starter, and $P\glue Q\cong P$ as
  posets if and only if $Q$ is a terminator.
\end{lemma}

\begin{proof}
  We only show the first claim; the second one follows by opposition.
  If $P$ is a starter, then obviously $P\glue Q\cong Q$ as posets.
  Conversely, if $P\glue Q\cong Q$ as posets, then
  $(P\sqcup Q)_{/t_P(i)= s_Q(i)}=Q$.  Thus $t_P$ must be surjective
  and therefore bijective. \qed
\end{proof}

In light of Lemma \ref{le:trivial-deglue} we may now define when a
gluing (de)composi\-tion is non-trivial.  This is easy for
\emph{posets} where $P\pomser Q$ is trivial precisely when $P$ or $Q$
is empty; but the presence of interfaces complicates matters.

\begin{definition}
  A gluing product $P\glue Q$ is \emph{trivial} if $P$ is a starter or
  $Q$ a terminator.
\end{definition}

Every identity is obviously a starter and a terminator, but not every
iposet that is a starter and a terminator is an identity.
A \emph{symmetry} is a discrete iposet $(s,P,t):n\to n$ that is both a
starter and a terminator. Thus $s$ and $t$ are both bijections.  All
points of $P$ are in the starting and terminating interfaces, but
$t^{-1}\circ s:[n]\to [n]$ is  merely a permutation, not necessarily an
identity. Intuitively, symmetries permute the  numbers that label
interfaces. The following lemma explains their name. 

\begin{lemma}
  \label{le:symm-symm}
  For any $n\ge 0$, the automorphisms on $[n]$ are precisely the
  symmetries $n\to n$ in $\iPos$.
\end{lemma}

\begin{proof}
  Any symmetry $n\to n$ is obviously an isomorphism on $[n]$.  For the
  converse, let $(s, P, t): n\to n$ be an isomorphism and
  $(s_1, Q_1, t_1): n\to n$ and $(s_2, Q_2, t_2): n\to n$ its left and
  right inverses.  The numbers of points in $Q_1$ and $P$ satisfy
  $|Q_1|\ge |T_{Q_1}|=n$ and $|P|\ge |S_P|=n$, and then
  $n=|Q_1\glue P|=|Q_1|+|P|-n$ implies $|P|=n$, so that $P$ must be a
  terminator.  A symmetric argument using $Q_2$ shows that $P$ is a
  starter, hence a symmetry. \qed
\end{proof}

The next lemma shows that parallel composition of iposets commutes up
to symmetry.

\begin{lemma}
  \label{le:symmetries}
  For any iposets $P_1:n_1\to m_1$ and $P_2:n_2\to m_2$ there are
  symmetries $\sigma$ on $n_1+n_2$ and $\tau$ on $m_1+m_2$ such that
  $P_1\otimes P_2\cong \sigma\glue(P_2\otimes P_1)\glue \tau$.
\end{lemma}

\begin{proof}
  In light of Lemma \ref{le:symm-symm}, the symmetries may be defined by
  \begin{equation*}
    \sigma(i)=
    \begin{cases}
      i+n_2 &\text{for } i\le n_1,\\
      i-n_1 &\text{for } i>n_1
    \end{cases}
    \qquad\text{ and }\qquad
    \tau(i)=
    \begin{cases}
      i+m_1 &\text{for } i\le m_2,\\
      i-m_2 &\text{for } i>m_2.
    \end{cases}
  \end{equation*}

  \vspace*{-6ex}
  \qed
\end{proof}

Non-identity symmetries may be removed from our setting by imposing
interface consistency, which we define next.

The interfaces of an iposet $(s,P,t):n\to m$ induce implicit extra
orderings on some of the points of $P$ that are independent of the
order on $P$. They are defined by $x\intord_s y$ if $x,y\in S_P$ and
$s^{-1}(x) <_\Nat s^{-1}(y)$, and $x\intord_t y$ if $x,y\in T_P$ and
$t^{-1}(x) <_\Nat t^{-1}(y)$.  Here $<_\Nat$ is the natural ordering
on $[n]$ and $[m]$.

\begin{definition}\label{D:int-cons}
  Iposet $(s,P,t):n\to m$ is \emph{interface consistent} if
  \begin{equation*}
    s^{-1}(x)<_\Nat s^{-1}(y) \,\Leftrightarrow\, t^{-1}(x)<_\Nat t^{-1}(y)\qquad \text{ for
      all } x,y\in S_P\cap T_P.
    \end{equation*}
\end{definition}

The orders $\intord_s$ and $\intord_t$ of interface consistent iposets
can therefore be combined into a partial order
$\intord\ =\ \intord_s\cup \intord_t$ on $P$. The interface consistent
symmetries are precisely the identities, as all points of a symmetry
are in $S_P\cap T_P$. Further, interface consistency is preserved by
gluing composition, parallel composition and subsumption.  Interface
consistent iposets thus form a subcategory of $\iPos$.  By Lemma
\ref{le:symm-symm}, all its automorphisms are identities.

We will see in Lemma \ref{le:gp-cons} that all iposets generated from
singletons using finitary gluing and parallel compositions are
interface consistent.  On the other hand, there are interface
consistent iposets which are not gluing-parallel, see Example
\ref{ex:non-gp-posets-parallel} below.

\section{A Criterion for Gluing Decompositions}
\label{sec:criterion}

In this section we supply a criterion for the existence of gluing
decompositions in iposets.  When thinking of a decomposition
$P=Q\glue R$ as synchronic cut through $P$, then some events are in
the past, already terminated, some in the present, currently running,
and some in the future, yet to be started.  This is captured in the
following definition.

\begin{definition}\label{D:char-fun}
  Let $Q:n\to m$ and $R:m\to k$ be iposets. The \emph{characteristic
    function} of the decomposition $Q\glue R$ is
  $\phi_{Q\glue R}:Q\glue R\to\{0,{*},1\}$ defined by
  \begin{equation*}
    \phi_{Q\glue R}(x) =
    \begin{cases}
      1 & \text{for } x\in Q\setminus T_Q, \\
      * & \text{for }x\in T_Q=S_R, \\
      0 & \text{for }x\in R\setminus S_R.
    \end{cases}
  \end{equation*}
\end{definition}

Hence we label past events with respect to the decomposition with $1$,
present events with $*$, and future events with $0$.

\begin{lemma}
  \label{le:characteristic-function}
  The characteristic function $\phi=\phi_{Q\glue R}$ satisfies the
  following:
  \begin{enumerate}[label=(\Alph*)]
  \item\label{le:characteristic-functionA}
    If $(\phi(x),\phi(y))=(1,0)$, then $x<y$.
  \item\label{le:characteristic-functionB}
    If $(\phi(x),\phi(y))\in\{(1,*),(*,0),(1,0)\}$, then $y\not<x$.
  \item\label{le:characteristic-functionC}
    If $(\phi(x),\phi(y))=(*,*)$, then $x\not<y$ and $y\not<x$.
  \item\label{le:characteristic-functionD}
    If $x<y$ and $\phi(y)\ne 0$, then $\phi(x)=1$.
    If $x<y$ and $\phi(x)\ne 1$, then $\phi(y)=0$.	
  \end{enumerate}
  If the decomposition $P=Q\glue R$ is non-trivial, then there exist
  $x,y\in P$ such that $x$ is minimal, $y$ is maximal, $\phi(x)=1$,
  and $\phi(y)=0$.
\end{lemma}

\begin{proof}
  \ref{le:characteristic-functionA} follows immediately from the
  definition of $\glue$. \ref{le:characteristic-functionB} holds
  because source (target) interfaces contain only minimal (maximal)
  elements and by \ref{le:characteristic-functionA}.
  \ref{le:characteristic-functionC} follows from incomparability of
  interface elements.  \ref{le:characteristic-functionD} follows from
  \ref{le:characteristic-functionB} and
  \ref{le:characteristic-functionC}.  The last statement is
  obvious. \qed
\end{proof}

Let $P$ be a poset, and let $P_a, P_b\subseteq P$ be the sets of
points $x$ for which the up-sets $x\up=\{y\mid x<y\}$ and down-sets
$x\down=\{y\mid y< x\}$ have maximal size, 
\begin{equation*}
  P_a = \big\{ x\in P\mid \forall y\in P.\ |y\up|\le|x\up| \big\},
  \qquad%
  P_b = \big\{ x\in P\mid \forall y\in P.\ |y\down|\le|x\down|
  \big\}.
\end{equation*}
Then $P_a\ne \emptyset\ne P_b$, every element of $P_a$ is minimal in
$P$ and every element of $P_b$ is maximal in $P$.

\begin{lemma}
  \label{le:criterion-for-gluing}
  Suppose $P$ admits a non-trivial gluing decomposition.  Then there
  is $\phi:P\to \{0,*,1\}$ that satisfies the conditions
  of Lemma \ref{le:characteristic-function} and, in addition,
  \begin{enumerate}[label=(\Alph*),start=5]
  \item\label{le:criterion-for-gluingE} $\phi(x)=1$ for $x\in P_a$ and
    $\phi(y)=0$ for $y\in P_b$.
  \end{enumerate}
\end{lemma}

\begin{proof}
  We only verify \ref{le:criterion-for-gluingE}.  Suppose $P=Q\glue R$
  is non-trivial. Then there must be $y,z\in P$ such that
  $\phi_{Q\glue R}(y)=1$, $\phi_{Q\glue R}(z)=0$ and therefore $y<z$.

  Let $x\in P_a$. If $\phi_{Q\glue R}(x)=0$, then $x>y$ by
  \ref{le:characteristic-functionA}, which contradicts minimality of
  $x$.  Suppose $\phi_{Q\glue R}(x)=*$, then
  $x\up\subseteq R\setminus S_R\subseteq y\up$.  But
  $|x\up|\ge|y\up|$, forcing $x\up=y\up$. In other words, $x<w$ for
  any $w\in R\setminus S_R$.  Hence there is a non-trivial gluing
  decomposition $P=Q'\glue R'$ in which $x$ has been moved from the
  gluing interface into $Q'\setminus T_{Q'}$ and consequently
  $\phi_{Q'\glue R'}(x)=1$.

  The proof of $\phi_{Q\glue R}(x)=0$ for $x\in P_b$ follows by
  opposition. \qed
\end{proof}

The following is now immediate from \ref{le:characteristic-functionA}
and \ref{le:criterion-for-gluingE}.

\begin{corollary}
  \label{co:criterion-simple}
  If poset $P$ admits a non-trivial gluing decomposition, then $x<y$
  for all $x\in P_a$ and $y\in P_b$. \qed
\end{corollary}

\section{Gluing-Parallel Iposets}
\label{S:gp}

\begin{figure}[tbp]
  \centering
  \begin{tikzpicture}
    \begin{scope}
      \node at (0,0) {$[0]\to [1]\from [0]$};
      \node at (3,0) {$[0]\to [1]\from [1]$};
      \node at (6,0) {$[1]\to [1]\from [0]$};
      \node at (9,0) {$[1]\to [1]\from [1]$};
    \end{scope}
    \begin{scope}[shift={(0,-1)}]
      \node at (0,0) {\intpt};
      \node [label=right:{\tiny $1$}] at (3,0) {\outpt};
      \node [label=left:{\tiny $1$}] at (6,0) {\inpt};
      \node [label=left:{\tiny $1$}] at (8.94,0) {\inpt};
      \node [label=right:{\tiny $1$}] at (9.06,0) {\outpt};
    \end{scope}
  \end{tikzpicture}
  \caption{The four singletons, structurally and graphically.}
  \label{fi:singletons}
\end{figure}
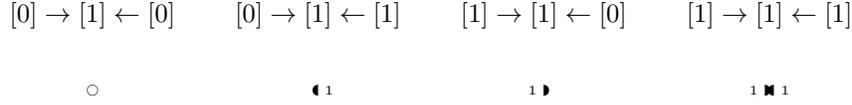

In this section, we start our study of iposets generated from
singletons using $\glue$ and $\otimes$.
A \emph{singleton} is an iposet whose underlying poset has one single
point.  There are four of them: $\rintpt$, $\routpt\! \rtone$,
$\rtone\! \rinpt$ and $\rtone\! \riopt\! \rtone$.  In particular,
$\id_1= \rtone\!\riopt\! \rtone$.  Figure \ref{fi:singletons} shows
their structural definitions and graphical representations. We write
$\single4=\{\rintpt, \routpt\! \rtone, \rtone\! \rinpt, \rtone\!
\riopt\! \rtone\}$ for the set of singleton iposets.  For
$i,j\in \{0,1\}$, we write more generally $\rti\rintpt\rtj$ instead of
$(f_i,[1],f_j):n\to m$, where $f_k: [k]\to [1]$ is either empty or the
identity.

\begin{definition}
  The set of \emph{gluing-parallel} iposets (\emph{gp-iposets}) is the
  smallest set that contains the empty iposet $\id_0$ and all elements
  of $\single4$ and  is closed under gluing and parallel
  composition.
\end{definition}

\begin{lemma}
  \label{le:gp-cons}
  Every gp-iposet is interface consistent.
\end{lemma}

\begin{proof}
  The empty iposet and all singletons are interface consistent; gluing
  and parallel compositions preserve this property. \qed
\end{proof}

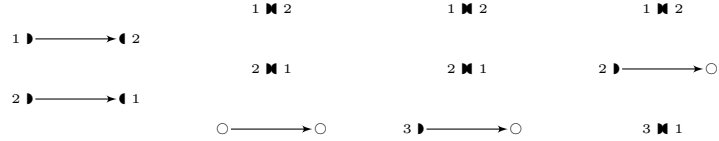
\begin{figure}[tbp]
  \centering
  \begin{tikzpicture}[x=1.3cm, y=.8cm]
    \begin{scope}[shift={(0,-.5)}]
      \node [label=left:{\tiny $1$}] (1) at (0,0) {\inpt};
      \node [label=left:{\tiny $2$}] (2) at (0,-1) {\inpt};
      \node [label=right:{\tiny $1$}] (3) at (1,-1) {\outpt};
      \node [label=right:{\tiny $2$}] (4) at (1,0) {\outpt};
      \path (1) edge (4);
      \path (2) edge (3);
    \end{scope}
    \begin{scope}[shift={(2,0)}]
      \node [label=left:{\tiny $1$}] at (.45,0) {\inpt};
      \node [label=right:{\tiny $2$}] at (.55,0) {\outpt};
      \node [label=left:{\tiny $2$}] at (.45,-1) {\inpt};
      \node [label=right:{\tiny $1$}] at (.55,-1) {\outpt};
      \node (1) at (0,-2) {\intpt};
      \node (2) at (1,-2) {\intpt};
      \path (1) edge (2);
    \end{scope}
    \begin{scope}[shift={(4,0)}]
      \node [label=left:{\tiny $1$}] at (.45,0) {\inpt};
      \node [label=right:{\tiny $2$}] at (.55,0) {\outpt};
      \node [label=left:{\tiny $2$}] at (.45,-1) {\inpt};
      \node [label=right:{\tiny $1$}] at (.55,-1) {\outpt};
      \node [label=left:{\tiny $3$}] (1) at (0,-2) {\inpt};
      \node (2) at (1,-2) {\intpt};
      \path (1) edge (2);
    \end{scope}
    \begin{scope}[shift={(6,0)}]
      \node [label=left:{\tiny $1$}] at (.45,0) {\inpt};
      \node [label=right:{\tiny $2$}] at (.55,0) {\outpt};
      \node [label=left:{\tiny $2$}] (1) at (0,-1) {\inpt};
      \node (2) at (1,-1) {\intpt};
      \node [label=left:{\tiny $3$}] at (.45,-2) {\inpt};
      \node [label=right:{\tiny $1$}] at (.55,-2) {\outpt};
      \path (1) edge (2);
    \end{scope}
  \end{tikzpicture}
  \caption{Four iposets on four points that are not gluing-parallel.}
  \label{fi:ipng4}
\end{figure}

\begin{example}
  \label{ex:non-gp-posets-parallel}
  Interface consistency characterises gp-iposets on two and three
  points: on two points, the only iposet that is not in this class is
  the non-trivial symmetry on $[2]$; on three points, twelve iposets
  are not in it. All of them are discrete and not interface
  consistent.

  On four points, there are $113$ non-gp iposets; $96$ of them are
  discrete (and not interface consistent).  Of the seventeen others,
  sixteen are parallel products of the non-trivial symmetry on $[2]$
  with an arrow.  The last one is the $\twotwo$ with interfaces
  swapped, which is the only interface consistent iposet on four
  points which is not gluing-parallel.  In Figure \ref{fi:ipng4}, the
  latter is displayed on the left, followed by three examples of the
  former.
\end{example}

Recall that the \emph{series-parallel} posets are freely generated
from the empty poset and the singleton $\rintpt$ by finitary serial
and parallel compositions and up to unicity and associativity of these
compositions as well as commutativity of parallel product
\cite{DBLP:journals/tcs/BloomE96a}.  There remain several obstructions
to showing a similar result for gp-iposets.

First, the lax interchange law from Proposition
\ref{P:lax-interchange} becomes strict if two of the components
involved are singletons with matching interfaces.

\begin{proposition}[Singleton interchange]
  \label{P:interchange-eq}
  For all iposets $P$, $Q$ with $P\glue Q$ defined and
  $i,j\in \{ 0, 1\}$,
  \begin{align*}
    ({\rti\rintpt\rtone} \otimes P) \glue ({\rtone\rintpt\rtj} \otimes Q) &\cong (
    {\rti\rintpt\rtone} \glue {\rtone\rintpt\rtj} ) \otimes (P\glue Q), \\
    (P\otimes {\rti\rintpt\rtone} ) \glue (Q\otimes {\rtone\rintpt\rtj}) &\cong (P\glue
    Q) \otimes ({\rti\rintpt\rtone} \glue {\rtone\rintpt\rtj}).
  \end{align*}
\end{proposition}

\begin{proof}
  As $\rti\rintpt\rtone$ and $\rtone\rintpt\rtj$ are both
  single points in an interface, the result follows by definition of
  $\glue$: the singletons glue separately in the left-hand sides of
  the identities precisely in the way described in their right-hand
  sides.

  For a detailed proof, first note that
  $\rti\rintpt\rtone \;\glue\; \rtone\rintpt\rtj = \rti\rintpt\rtj$. By Proposition
  \ref{P:lax-interchange}, there is a subsumption
  $f:({\rti\rintpt\rtone}\otimes P)\glue({\rtone\rintpt\rtj}\otimes Q)\to
  {\rti\rintpt\rtj}\otimes( P\glue Q) = L\to R$, which is the identity
  on the two underlying posets.  It remains to show that it preserves
  the order.  So let $x, y\in L$ with $x<_L y$.  If
  $x, y\in {\rti\rintpt\rtone}\otimes P$, then $x<_P y$, hence
  $x<_{P\glue Q} y$ and $x<_R y$, and likewise for
  $x, y\in {\rtone\rintpt\rtj}\otimes Q$.  Otherwise, if
  $x\in {\rti\rintpt\rtone} \otimes P\setminus T_{{\rti\rintpt\rtone} \otimes P}$ and
  $y\in {\rtone\rintpt\rtj}\otimes Q\setminus S_{{\rtone\rintpt\rtj}\otimes
    Q}$, then $x\in P\setminus T_P$ and $y\in Q\setminus S_Q$, thus
  $x<_{P\glue Q} y$ and again $x<_R y$. \qed
\end{proof}

Second, parallel composition becomes commutative when some components
have no interfaces.

\begin{proposition}
  \label{P:otimes-comm}
  Let $P_1:n_1\to m_1$ and $P_2:n_2\to m_2$ be iposets and assume that
  $n_1=0$ or $n_2=0$, and $m_1=0$ or $m_2=0$.  Then
  $P_1\otimes P_2\cong P_2\otimes P_1$.
\end{proposition}
\begin{proof}
  Using Lemma \ref{le:symmetries} there are symmetries $\sigma$ and
  $\tau$ such that
  $P_1\otimes P_2\cong \sigma\glue(P_2\otimes P_1)\glue \tau$; but the
  assumptions make both $\sigma$ and $\tau$ identities. \qed
\end{proof}

Propositions~\ref{P:interchange-eq} and \ref{P:otimes-comm} have
converses, which tell us precisely when strict interchange and
commutativity hold.  First, by the next lemma, Proposition
\ref{P:interchange-eq} covers all cases of strict interchange.

\begin{lemma}
  \label{le:decomp}
  Let $P_1\otimes P_2\cong Q_1\glue Q_2$ such that the gluing
  composition is non-trivial.  Then $P_1$ or $P_2$ is discrete.
\end{lemma}

\begin{figure}[tbp]
  \centering
  \begin{tikzpicture}[-, shorten <=0pt, shorten >=0pt, x=1.7cm]
    \draw (0,0) -- (2,0) -- (2,2) -- (0,2) -- (0,0);
    \draw (1,0) -- (1,2) (0,1) -- (2,1);
    \node at (-.3,.5) {$P_2$};
    \node at (-.3,1.5) {$P_1$};
    \node at (.5,-.4) {$Q_1$};
    \node at (1.5,-.4) {$Q_2$};
    \node at (.5,.5) {$P_{21}$};
    \node at (.5,1.5) {$P_{11}$};
    \node at (1.5,.5) {$P_{22}$};
    \node at (1.5,1.5) {$P_{12}$};
  \end{tikzpicture}
  \caption{Partition of $P_1\otimes P_2\cong Q_1\glue Q_2$ in the
    proof of Lemma \ref{le:decomp}.}
  \label{fi:pqinterchange}
\end{figure}

\begin{proof}
  Partition each $P_i$ into $P_{ij}$ containing the points in $P_i$ that
  are also in $Q_j$, as indicated in Figure \ref{fi:pqinterchange}.
  By hypothesis, there is a non-target point $p\in Q_1\setminus T_{Q_1}$ and a
  non-source point $q\in Q_2\setminus S_{Q_2}$, hence $p<q$ in $Q_1\glue Q_2$.
  More precisely, we must have $p\in P_{i1}$ and $q\in P_{i2}$ for $i$
  either $1$ or $2$, because $P_1$ and $P_2$ are $<$-disconnected.
  Now any non-terminating point $r\in P_{j1}$ for $j\ne i$ would force
  $r<q$ in $Q_1\glue Q_2$, which is inconsistent with this
  disconnectivity.  By opposition, any non-starting point
  $r\in P_{j2}$ would force $p<r$ in $Q_1\glue Q_2$.  Hence $P_{j1}$
  must be a starter and $P_{j2}$ a terminator, making $P_j$
  discrete. \qed
\end{proof}

Let $\sim$ denote the equivalence relation generated by $<$ in a poset
$P$.  Equivalence classes of $\sim$ are \emph{connected components} of
$P$.  $P$ is \emph{connected} if it has exactly one connected
component.

\begin{lemma}
  \label{le:parcom-nointf}
  Let $P_1:n_1\to m_1$ and $P_2:n_2\to m_2$ be iposets such that $P_1$
  and $P_2$ are both connected, $P_1\not\cong P_2$, and
  $P_1\otimes P_2\cong P_2\otimes P_1$.  Then $n_1=0$ or $n_2=0$, and
  $m_1=0$ or $m_2=0$.
\end{lemma}

\begin{proof}
  Let $f: P_1\otimes P_2\to P_2\otimes P_1$ be any iposet
  isomorphism and suppose $n_1\ne 0\ne n_2$. We show that
  $f|_{P_1}:P_1\to P_2$ is an isomorphism, too.  The proof for
  $m_1\neq 0\neq m_2$ is then symmetric.

  Both $P_1\otimes P_2$ and $P_2\otimes P_1$ have two connected
  components: $[P_1]$ and $[P_2]$.  The condition $n_1\neq 0\neq n_2$
  guarantees that
  $f(s_{P_1\otimes P_2}(1)) = s_{P_2\otimes P_1}(1) \in P_2$ and
  $s_{P_1\otimes P_2}(1) \in P_1$.  Thus, $f$ sends the connected
  component $[P_1]\subseteq P_1\otimes P_2$ to
  $[P_2]\subseteq P_2\otimes P_1$.  Since $f$ is an isomorphism, its
  restriction $f|_{P_1}:P_1\to P_2$ is also an isomorphisms of posets.
  
  Let $s_{12}:[n_1+n_2]\to P_1\otimes P_2$,
  $s_{21}:[n_1+n_2]\to P_2\otimes P_1$ denote the respective source
  interfaces for short. By assumption, $f\circ s_{12}=s_{21}$.  Thus,
  for $i\in[n_1]$,
  \[
  	f(s_1(i))
  	=f((s_1\otimes s_2)\circ \varphi_{n_1,n_2}(i))
  	= f(s_{12}(i))
  	=s_{21}(i)
  	=(s_2\otimes s_1)\circ \varphi_{n_2,n_1}(i).
  \]
  We have $s_1(i)\in [P_1]\subseteq P_1\otimes P_2$
  and, therefore, $f(s_1(i))\in P_2\subseteq [P_2]$.
  Hence, $\varphi_{n_2,n_1}(i)\leq n_2$, \ie $i\leq n_2$.
  Finally,
  \[
  	f(s_1(i))=(s_2\otimes s_1)\circ \varphi_{n_2,n_1}(i) = s_2(i).
  \]
  Similarly we show that $f\circ t_1=t_2$.  This implies that
  $f|_{P_1}:P_1\to P_2$ preserves interfaces and is an isomorphism of
  iposets. \qed
\end{proof}

The connectedness assumption above is needed for rather trivial
reasons: if $P\ne \id_0$, then $P\ne P\otimes P$, but
$P\otimes(P\otimes P)= (P\otimes P)\otimes P$.

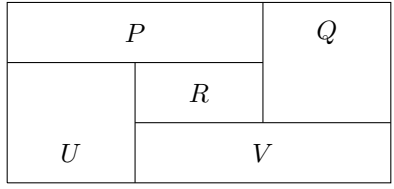
\begin{figure}[tbp]
  \centering
  \begin{tikzpicture}[-, shorten <=0pt, shorten >=0pt, x=1.7cm,
    y=.8cm]
    \draw (0,0) -- (3,0) -- (3,3) -- (0,3) -- (0,0);
    \draw (1,1) -- (3,1) (0,2) -- (2,2);
    \draw (1,0) -- (1,2) (2,1) -- (2,3);
    \node at (.5,.5) {$U$};
    \node at (2,.5) {$V$};
    \node at (1.5,1.5) {$R$};
    \node at (1,2.5) {$P$};
    \node at (2.5,2.5) {$Q$};
  \end{tikzpicture}
  \caption{One of the two situations in Levi's lemma.}
  \label{fi:levi}
\end{figure}

Whether gluing-parallel iposets form the free algebra in some variety
remains open.  Such a result is claimed in \cite[Thm.\
19]{DBLP:conf/RelMiCS/FahrenbergJST20}, but its proof depends on a
lemma, \cite[Lemma 16]{DBLP:conf/RelMiCS/FahrenbergJST20}, that does
not hold.  It states that if $P\glue Q=U\glue V$ for some iposets,
then these gluing decompositions have a common refinement, that is,
there exists an iposet $R$ such that either $P= U\glue R$ and
$R\glue Q= V$ or $U= P\glue R$ and $R\glue V= Q$. See Figure
\ref{fi:levi} for an illustration.  This claim is refuted by the
following simple counterexample.

\begin{example}
  \label{ex:antilevi}
  Consider the iposets
  \begin{equation*}
    P =\! \vcenter{\hbox{%
        \begin{tikzpicture}[y=.5cm]
          \node (0) at (0,0) {\intpt};
          \node (2) at (1,0) {\intpt};
          \node [label=right:{\tiny $1$}] (3) at (1,-1) {\outpt};
          \path (0) edge (2);
        \end{tikzpicture}}}, \qquad
    U =\! \vcenter{\hbox{%
        \begin{tikzpicture}[y=.5cm]
          \node (0) at (0,0) {\intpt};
          \node [label=right:{\tiny $1$}] (2) at (1,0) {\outpt};
          \node (3) at (1,-1) {\intpt};
          \path (0) edge (2);
        \end{tikzpicture}}}, \qquad
    V =\! \vcenter{\hbox{%
        \begin{tikzpicture}[y=.5cm]
          \node [label=left:{\tiny $1$}] (0) at (0,0) {\inpt};
          \node (2) at (1,0) {\intpt};
          \path (0) edge (2);
        \end{tikzpicture}}}.
  \end{equation*}
  Then
  \begin{equation*}
    P\glue V=U\glue V =\! \vcenter{\hbox{%
        \begin{tikzpicture}[y=.5cm]
          \node (0) at (0,0) {\intpt};
          \node (1) at (1,0) {\intpt};
          \node (2) at (2,0) {\intpt};
          \node (3) at (1,-1) {\intpt};
          \path (0) edge (1);
          \path (1) edge (2);
          \path (3) edge (2);
        \end{tikzpicture}}},
  \end{equation*}
  but neither $P\glue R=U$ nor $U\glue R=P$ for any iposet $R$.
\end{example}

\section{A Hierarchy of Gluing-Parallel Iposets}
\label{se:hierarchy}

In order to refine our study of gluing-parallel iposets, and in
particular to work towards a result on free generation as mentioned
above, we now define a hierarchy of gp-iposets generated from the
empty iposet and the singletons in $\single4$ by successive iterations
of parallel and gluing compositions.  We show that this hierarchy is
infinite, which is a necessary, but not a sufficient condition for
gp-iposets being freely generated.

For any $\mcal Q\subseteq \iPos$ and
$\mathord{\Box} \in \{\mathord{\otimes}, \mathord{\glue}\}$ let
\begin{equation*}
  \closure{\Box}{\mcal Q} = \{P_1\mathop{\Box}\dotsm\mathop{\Box}
  P_n\mid n\ge 1, P_1,\dotsc,P_n\in \mcal Q\}
\end{equation*}
be the $\Box$-closure of $\mcal Q$. In particular, therefore,
$\mcal Q\subseteq \closure{\Box}{\mcal Q}$.  This defines operations
$\closureop{\glue},\closureop{\otimes}: 2^\iPos\to 2^\iPos$.  We are
interested in their iterative application and write
\begin{equation*}
  \single4_0= \{\id_0\}\cup \single4,\qquad \single4_{i+1} =
  \closureop{\glue\otimes} \single4_i,\qquad \single4_\ast = \bigcup_{i\ge 0} \single4_i.
\end{equation*}
Then $\single4_\ast$ is the class of gp-iposets.

Further, we are interested in gluing-parallel \emph{posets}, which we
regard as iposets with empty interfaces.  To this end, let
\begin{equation*}
  \mcal G= \single4\cap \Pos,\qquad \mcal G_n=\mcal
  S_n\cap \Pos, \qquad {\mcal G}_\ast= {\mcal S}_\ast\cap \Pos.
  \end{equation*}

\begin{lemma}
\label{le:ICDiscrete}
  An iposet is in
  $\closure{\otimes}{\single4}$ if and only if it is discrete and
  interface consistent.
\end{lemma}

\begin{proof}
  Any iposet in $\closure{\otimes}{\single4}$ is obviously discrete
  and interface consistent.  Conversely, let $(s,P,t):n\to m$ be
  discrete and interface consistent.  We extend the partial order
  $\intord$ defined after Definition~\ref{D:int-cons} to a linear
  order, also denoted $\intord$, on $P$.  If the components of $P$ are
  ordered so that $P_1\intord\dotsm\intord P_k$, then
  $P= P_1\otimes\dotsm\otimes P_k\in \closure{\otimes}{\single4}$,
  with $P_i \in\single4$ for $1\leq i \leq k$. \qed
\end{proof}

\begin{lemma}
  If $(s,P,t):n\to m$ is in ${\mcal S}_k$, then so is
  $(\emptyset,P,\emptyset):0\to 0$, and $\mcal G_k$ is the set of
  underlying posets of iposets in $\mcal S_k$.
\end{lemma}
\begin{proof}
  This is straightforward for $k=0$.  For $k>0$, there is a
  presentation $(\emptyset,P,\emptyset)=S\glue (s,P,t)\glue T$, where
  $S:0\to n$ is a starter and $T:m\to 0$ is a terminator.  Since
  $S,T\in{\otimes}{\single4}\subseteq{\mcal S}_1$, the claim
  follows.\qed
\end{proof}

Next we show that the interval orders are precisely the posets in
$\mcal G_1$.  We use the following technical lemma, which uses a
suitable finite linear order as the image of an interval
representation instead of $\Real$.

\begin{lemma}
  \label{le:IIRepresentation}
  Every iposet $(s,P,t)$ with $P$ an interval order has an interval
  representation $(b,e):P\to Q$ into a finite linear order $Q$ such
  that $x\in S_P$ if and only if $b(x)$ is the least element $\bot$ of
  $Q$ and $x\in T_P$ if and only if $e(x)$ is the greatest element
  $\top$ of $Q$.  Moreover, we may assume that $|Q|\geq 3$ and
  $b(p)<\top$, $e(p)>\bot$ for all $p\in P$.
\end{lemma}

\begin{proof}
  For any interval presentation $(b',e'):P\to Q'$, let $Q$ be the linear
  order obtained by adjoining a new minimal element $\bot$ and a new
  maximal element $\top$ to $Q'$. Then
	\[
		b(x)=
		\begin{cases}
			b'(x) & \text{for }x\not\in S_P,\\
			\bot & \text{for }x\in S_P
		\end{cases}
		\qquad\text{ and }\qquad
		e(x)=
		\begin{cases}
			e'(x) & \text{for }x\not\in T_P,\\
			\top & \text{for }x\in T_P
		\end{cases}
	\]
	define an interval presentation $(b,e):P\to Q$ satisfying the
        claim.\qed
\end{proof}

\begin{proposition}
  \label{P:iorder}
  The class $\mcal S_1$ is equal to the class of interface consistent
  interval orders.  The class $\mcal G_1$ is equal to the class of
  interval orders.
\end{proposition}

\begin{figure}[tbp]
  \centering
  \begin{tikzpicture}
    \begin{scope}[xscale=.4, yscale=.75, shorten <=-2pt,
      shorten >=-2pt, label distance=-.2cm]
      \begin{scope}
        \node (1) at (1.5,0) {$\vphantom{bf}a$};
        \node [label=right:{\tiny $1$}] (2) at (4.5,0) {$\vphantom{bf}b$};
        \node (3) at (1.5,-1) {$\vphantom{bf}c$};
        \node (4) at (4.5,-1) {$\vphantom{bf}d$};
        \node [label=right:{\tiny $2$}] (5) at (4.5,-2) {$\vphantom{bf}e$};
        \foreach \i/\j in {1/2,3/2,3/4,3/5} \path (\i) edge (\j);
        \node at (6.5,-1) {$\glue$};
      \end{scope}
      \begin{scope}[xshift=7.5cm]
        \node [label=left:{\tiny $1$}] (6) at (.5,0) {$\vphantom{bf}f$};
        \node (7) at (3.5,0) {$\vphantom{bf}g$};
        \node [label=left:{\tiny $2$}] (8) at (.5,-2) {$\vphantom{bf}h$};
        \node (9) at (3.5,-2) {$\vphantom{bf}i$};
        \foreach \i/\j in {6/7,8/7,8/9} \path (\i) edge (\j);
        \node at  (5.7,-1) {$=$};
      \end{scope}
      \begin{scope}[xshift=16cm]
        \node (1) at (-.5,0) {$\vphantom{bf}a$};
        \node (2) at (4.5,0) {$\vphantom{bf}bf$};
        \node (3) at (-.5,-1) {$\vphantom{bf}c$};
        \node (4) at (2,-1) {$\vphantom{bf}d$};
        \node (5) at (4.5,-2) {$\vphantom{bf}eh$};
      \end{scope}
      \begin{scope}[xshift=20cm]
        \node (6) at (.5,0) {\phantom{$\vphantom{bf}bf$}};
        \node (7) at (5.5,0) {$\vphantom{bf}g$};
        \node (8) at (.5,-2) {\phantom{$\vphantom{bf}eh$}};
        \node (9) at (5.5,-2) {$\vphantom{bf}i$};
      \end{scope}
      \foreach \i/\j in {1/2,3/2,3/4,3/5} \path (\i) edge (\j);
      \foreach \i/\j in {6/7,8/7,8/9} \path (\i) edge (\j);
      \foreach \i/\j in {1/9,4/7,4/9} \path (\i) edge (\j);
    \end{scope}
    \begin{scope}[-, xscale=.4, yshift=-2.7cm, shorten <=-4.5pt, shorten >=-4.5pt]]
      \begin{scope}
        \node (1l) at (0,0) {{\tiny $|$}};
        \node (1r) at (3,0) {{\tiny $|$}};
        \node (2l) at (4,0) {{\tiny $|$}};
        \node (2r) at (5,0) {{\tiny $|$}};
        \node (3l) at (0,-1) {{\tiny $|$}};
        \node (3r) at (1,-1) {{\tiny $|$}};
        \node (4l) at (2,-1) {{\tiny $|$}};
        \node (4r) at (4.5,-1) {{\tiny $|$}};
        \node (5l) at (2,-2) {{\tiny $|$}};
        \node (5r) at (5,-2) {{\tiny $|$}};
        \path (1l) edge node[above] {\small $\vphantom{bf}I(a)$}  (1r);
        \path (2l) edge node[above] {\small $\vphantom{bf}I(b)$}  (2r);
        \path (3l) edge node[above] {\small $\vphantom{bf}I(c)$}  (3r);
        \path (4l) edge node[above] {\small $\vphantom{bf}I(d)$}  (4r);
        \path (5l) edge node[above] {\small $\vphantom{bf}I(e)$}  (5r);
        \node at (6.2,-1) {$\glue$};
      \end{scope}
      \begin{scope}[xshift=7.4cm]
        \node (6l) at (0,0) {{\tiny $|$}};
        \node (6r) at (3,0) {{\tiny $|$}};
        \node (7l) at (4,0) {{\tiny $|$}};
        \node (7r) at (5,0) {{\tiny $|$}};
        \node (8l) at (0,-2) {{\tiny $|$}};
        \node (8r) at (1,-2) {{\tiny $|$}};
        \node (9l) at (2,-2) {{\tiny $|$}};
        \node (9r) at (5,-2) {{\tiny $|$}};
        \path (6l) edge node[above] {\small $\vphantom{bf}I(f)$}  (6r);
        \path (7l) edge node[above] {\small $\vphantom{bf}I(g)$}  (7r);
        \path (8l) edge node[above, pos=2.5] {\small $\vphantom{bf}I(h)$}  (8r);
        \path (9l) edge node[above] {\small $\vphantom{bf}I(i)$}  (9r);
        \node at  (6.2,-1) {$=$};
      \end{scope}
      \begin{scope}[xshift=14.8cm]
        \node (1l') at (0,0) {{\tiny $|$}};
        \node (1r') at (3,0) {{\tiny $|$}};
        \node (2/6l') at (4,0) {{\tiny $|$}};
        \node (3l') at (0,-1) {{\tiny $|$}};
        \node (3r') at (1,-1) {{\tiny $|$}};
        \node (4l') at (2,-1) {{\tiny $|$}};
        \node (4r') at (4.5,-1) {{\tiny $|$}};
        \node (5/8l') at (2,-2) {{\tiny $|$}};
      \end{scope}
      \begin{scope}[xshift=21cm]
        \node (2/6r') at (3,0) {{\tiny $|$}};
        \node (7l') at (4,0) {{\tiny $|$}};
        \node (7r') at (5,0) {{\tiny $|$}};
        \node (5/8r') at (1,-2) {{\tiny $|$}};
        \node (9l') at (2,-2) {{\tiny $|$}};
        \node (9r') at (5,-2) {{\tiny $|$}};
      \end{scope}
      \path (1l') edge node[above] {\small $\vphantom{bf}I(a)$}  (1r');
      \path (2/6l') edge node[above] {\small $\vphantom{bf}I(bf)$}  (2/6r');
      \path (3l') edge node[above] {\small $\vphantom{bf}I(c)$}  (3r');
      \path (4l') edge node[above] {\small $\vphantom{bf}I(d)$}  (4r');
      \path (5/8l') edge node[above, pos=.3] {\small $\vphantom{bf}I(eh)$}  (5/8r');
      \path (7l') edge node[above] {\small $\vphantom{bf}I(g)$}  (7r');
      \path (9l') edge node[above] {\small $\vphantom{bf}I(i)$}  (9r');
      \path [dashed] (5,-0.2) edge (5,-2);
      \path [dashed] (7.4,-0.2) edge (7.4,-2);
      \path [dashed] (19.8,-0.2) edge (19.8,-2);
    \end{scope}
  \end{tikzpicture}
  \caption{Two interval orders and their concatenation: above as
    iposets, below using their interval representations.  (Labels
    added for convenience.)}
  \label{fi:intcomp}
\end{figure}
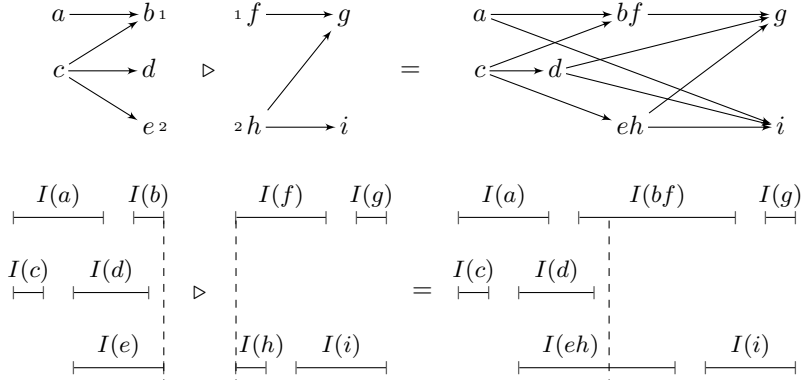

\begin{proof}
  It suffices to prove the first claim; the second one follows.  The
  underlying posets of all elements of $\closure{\otimes}{\single4}$
  are interval orders. To prove the forward implication, it thus
  suffices to show that the gluing of two interval orders yields an
  interval order.  Suppose $(b_P,e_P):P\to R$, $(b_Q,e_Q):Q\to S$ are
  interval representations satisfying the conditions of Lemma
  \ref{le:IIRepresentation}.  Then
  \begin{equation*}
    b(x)=
    \begin{cases}
      b_P(x) & \text{for }x\in P,\\
      b_Q(x) & \text{for }x\not\in P
    \end{cases}
    \qquad\text{ and }\qquad
    e(x)=
    \begin{cases}
      e_P(x) & \text{for }x\not\in Q,\\
      e_Q(x) & \text{for }x\in Q
    \end{cases}
  \end{equation*}
  define an interval representation
  $P\glue Q\to (R\sqcup S)/_{\top\!_R=\bot_S}$.  Figure \ref{fi:intcomp}
  shows an example.

  Conversely, suppose iposet $(s,P,t)$ is an interface consistent
  interval order with an interval representation
  $b,e:P\to Q=\{\bot<q_1<\dotsm<q_n<\top\}$ satisfying the conditions
  of Lemma \ref{le:IIRepresentation}.  We prove $P\in \mcal{S}_1$ by
  induction on $n$.  If $n=1$, then $b(p)\leq q_1\leq e(p')$ for all
  $p,p'\in P$.  Thus, $P$ is discrete and interface consistent and by
  Lemma \ref{le:ICDiscrete} belongs to
  ${\otimes}\single4\subseteq\single4_1$.
 
  Assume $n>1$.  Extend the order $\intord$ introduced below
  Definition \ref{D:int-cons} to an arbitrary linear order, also
  denoted $\intord$.
  Let $P'=\{p\in P\;|\; b(p)\leq q_1\}$,
  $P''=\{p\in P\;|\; e(p)\geq q_2\}$ be induced subposets of $P$.
  The intersection $P'\cap P''$ is a discrete poset.
  Let $m=|P'\cap P''|$ and $u:[m]\to (P'\cap P'',{\intord})$
  be the unique order-preserving function.
  Both iposets $(s,P',u)$ and $(u,P'',t)$ are interface consistent interval iposets,
  and both have shorter interval representations than $P$,
  namely,
  $b',e':P'\to\{\bot<q_1<\top\}$ and $b'',e'':P''\to\{\bot<q_2<\dots<q_n<\top\}$
  given by $b'(p)=b(p)$, $e''(p)=e(p)$ and
  \[
  	e'(p)=
  	\begin{cases}
  		e(p) & \text{for $e(p)\leq q_1$}\\
  		\top & \text{for $e(p)> q_1$}
  	\end{cases},
  	\qquad
  	b''(p)=
  	\begin{cases}
  		b(p) & \text{for $b(p)\leq q_2$}\\
  		\bot & \text{for $b(p)> q_2$}
  	\end{cases}.
  \]
  Hence, by inductive hypothesis, $P',P''\in \mcal{S}_1$.
  It is elementary to verify that  $(s,P,t)=(s,P',u)\glue (u,P'',t)$,
  which implies that $(s,P,t)\in \mcal{S}_1$. \qed

\end{proof}

To compare gluing-parallel posets with series-parallel ones we
construct a similar hierarchy for them.  Let
$\single1_0=\mcal G_0=\{\id_0,\rintpt\}\subset \single4_0$ be the set
containing the empty poset and the unique singleton without interfaces
and, like for the $\single4$ hierarchy above,
$\single1_{n+1}=\closureop{\glue\otimes} \single1_n$, and
$\single1_\ast =\bigcup_{n\ge 0} \single1_n$.

\begin{lemma}
  A poset is series-parallel if and only if it is in
  $\single1_\ast$.
\end{lemma}

\begin{proof}
  The elements of $\single1_0$ are the empty poset and the singleton
  with empty interfaces, and if every element of any
  $\mcal Q\subseteq \iPos$ has empty interfaces, then the same holds
  for $\closure{\otimes}{\mcal Q}$ and $\closure{\glue}{\mcal Q}$.
  Thus any element of any $\mcal T_n$ has empty interfaces.  Further,
  if $P$ and $Q$ have empty interfaces, then $P\glue Q$ is a serial
  composition.  The claim then follows. \qed
\end{proof}

\begin{lemma}
  $\single1_n\subset \mcal G_n$ for all $n\ge 1$, and
  $\single1_\ast\subset {\mcal G}_\ast$.
\end{lemma}

\begin{proof}
  $\single1_0\subset \single4_0$ implies
  $\single1_n\subseteq \single4_n$. Hence
  $\single1_n=\single1_n\cap \Pos\subseteq \single4_n\cap \Pos=\mcal
  G_n$, for all $n\ge 0$.  Thus also
  $\single1_\ast\subseteq \mcal G_\ast$, and the inequalities follow from
  $\single1_\ast\not\ni \N\in \mcal G_1$. \qed
\end{proof}

For an analogue of Proposition \ref{P:iorder} for sp-posets, recall
\cite{DBLP:books/sp/Vogler92} that a poset $P$ is a \emph{step
  sequence} if its incomparability relation $\incomp$ defined by
$x\incomp y \Leftrightarrow x\not<y \land y\not<x$ is transitive, so
that every point belongs to a unique $\incomp$-equivalence class.

\begin{proposition}
  A poset is in $\single1_1$ if and only if it is a step sequence.
\end{proposition}

\begin{proof}
  If $P=P_1\glue\dotsm\glue P_n$ for
  $P_1,\dotsc,P_n\in \closure{\otimes}{\single1_0}$, then each $P_i$
  is discrete, hence $\{P_1,\dotsc,P_n\}$ is the $\incomp$-partition
  of $P$.  Conversely, if $\incomp$ is transitive, then the
  $\incomp$-partition of $P$ can be totally ordered as
  $P_1<\dotsm<P_n$, and then $P=P_1\glue\dotsm\glue P_n$. \qed
\end{proof}

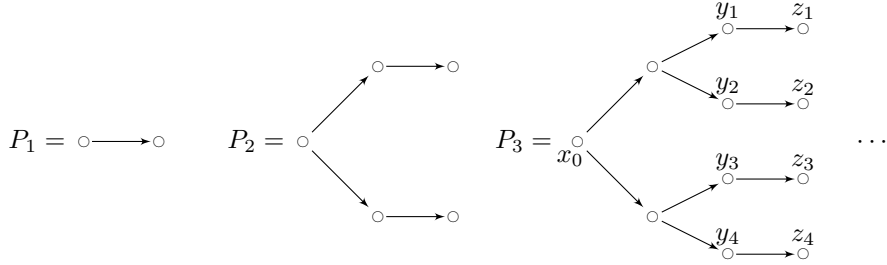
\begin{figure}[tbp]
  \begin{equation*}
    P_1=\! \vcenter{\hbox{%
        \begin{tikzpicture}
          \node (0) at (0,0) {\intpt};
          \node (1) at (1,0) {\intpt};
          \path (0) edge (1);
        \end{tikzpicture}}}
    \qquad%
    P_2=\! \vcenter{\hbox{%
        \begin{tikzpicture}
          \begin{scope}[shift={(1,1)}]
            \node (0a) at (0,0) {\intpt};
            \node (1a) at (1,0) {\intpt};
            \path (0a) edge (1a);
          \end{scope}
          \begin{scope}[shift={(1,-1)}]
            \node (0b) at (0,0) {\intpt};
            \node (1b) at (1,0) {\intpt};
            \path (0b) edge (1b);
          \end{scope}
          \node (0) at (0,0) {\intpt};
          \path (0) edge (0a);
          \path (0) edge (0b);
        \end{tikzpicture}}}
    \quad%
    P_3=\! \vcenter{\hbox{%
        \begin{tikzpicture}
          \path[use as bounding box] (-.3,0) to (3.3,0);
          \begin{scope}[yscale=.5, shift={(1,2)}]
            \begin{scope}[shift={(1,1)}]
              \node (0a) at (0,0) {\intpt};
              \node (1a) at (1,0) {\intpt};
              \path (0a) edge (1a);
              \node[above] at (1a) {$z_1$};
              \node[above] at (0a) {$y_1$};
            \end{scope}
            \begin{scope}[shift={(1,-1)}]
              \node (0b) at (0,0) {\intpt};
              \node (1b) at (1,0) {\intpt};
              \path (0b) edge (1b);
              \node[above] at (1b) {$z_2$};
              \node[above] at (0b) {$y_2$};
            \end{scope}
            \node (00a) at (0,0) {\intpt};
            \path (00a) edge (0a);
            \path (00a) edge (0b);
          \end{scope}
          \begin{scope}[yscale=.5, shift={(1,-2)}]
            \begin{scope}[shift={(1,1)}]
              \node (0a) at (0,0) {\intpt};
              \node (1a) at (1,0) {\intpt};
              \path (0a) edge (1a);
              \node[above] at (1a) {$z_3$};
              \node[above] at (0a) {$y_3$};
            \end{scope}
            \begin{scope}[shift={(1,-1)}]
              \node (0b) at (0,0) {\intpt};
              \node (1b) at (1,0) {\intpt};
              \path (0b) edge (1b);
              \node[above] at (1b) {$z_4$};
              \node[above] at (0b) {$y_4$};
            \end{scope}
            \node (00b) at (0,0) {\intpt};
            \path (00b) edge (0a);
            \path (00b) edge (0b);
          \end{scope}
          \node (0) at (0,0) {\intpt};
          \path (0) edge (00a);
          \path (0) edge (00b);
          \node[below] at (0) {$x_0$\;\;};
        \end{tikzpicture}}}
    \quad \dotsm
  \end{equation*}
  \caption{Sequence of separators for $\mcal T_n$ and
    $\mcal S_{n-1}$.}
  \label{fig:separators}
\end{figure}

Next we show that our three hierarchies are infinite, presenting a
sequence of witnesses for the strictness of inclusions.
Let $P_1=\rintpt\glue \rintpt$, and for $n\ge 1$,
$P_{n+1}=\rintpt\glue(P_n\otimes P_n)$.  See Figure
\ref{fig:separators} for a graphical representation. Note that all of
them are series-parallel posets.

\begin{lemma}
  \label{le:bintreedecomp}
  Let $P_{n+1}=Q\glue R$ be a non-trivial gluing decomposition for
  $n\geq 1$.  Then $R= P_n\otimes P_n$ as posets.
\end{lemma}

\begin{proof}
  Set $P=P_{n+1}$ and let $\phi=\phi_{Q\glue R}$ be the
  characteristic function of $Q\glue R$ (see
  Definition~\ref{D:char-fun}).  Then $P_a=\{x_0\}$ is the set of
  points for which the up-sets have maximal size because $x_0$ is the
  only minimal element of $P$, and $P_b=\{ z_i\mid i\in[2^n]\}$ is the
  set of points for which the down-sets have maximal size, because
  $|z_i\down|=n+1$ for all $i$.

  Lemma \ref{le:characteristic-function} implies that $\phi(x_0)=1$
  and $\phi(z_i)=0$ for some $i$.  For any
  $j\in[2^{n}]\setminus\{i\}$, $y_j\not< z_i$ and therefore
  $\phi(y_j)\ne 1$ by Lemma \ref{le:characteristic-function}\ref{le:characteristic-functionA}.  Then $\phi(z_j)=0$ because $y_j< z_j$, using
  Lemma
  \ref{le:characteristic-function}\ref{le:characteristic-functionD}.
  Now if $\phi(p)=1$, then $p<z_j$ for all $j$, which implies that
  $p=x_0$.  Finally,
  \begin{equation*}
    R=\phi^{-1}(\{*,0\})=P_{n+1}\setminus\{x_0\}=P_{n}\otimes P_{n}.
  \end{equation*}

  \vspace{-8ex}%
  \qed
\end{proof}

\begin{lemma}
  \label{le:bintree}
  $P_n\in \single1_n\setminus \mcal G_{n-1}$ for all $n\ge 1$.
\end{lemma}

\begin{proof}
  A simple induction shows that $P_n\in \mcal T_n$ for $n>1$. It
  remains to show that $P_n\notin \mcal G_{n-1}$ for all $n>1$. We
  proceed again by induction. In the base case, $P_1\notin \mcal
  G_0$. For the induction step, suppose $P_n\notin \mcal G_{n-1}$. We
  need to show that $P_{n+1}\notin \mcal G_n$. As $P_{n +1}$ is
  connected, it is of the form $Q\glue R$, and non-trivially
  so. Hence, by Lemma \ref{le:bintreedecomp}, $R=P_n\otimes P_n$. We
  know that $P_n\notin \mcal G_{n-1}$. It remains to show that
  $P_n\otimes P_n\notin \mcal G_{n-1}$. This follows from
  Corollary~\ref{co:criterion-simple}, as $P_a$ consists of the two
  minimal elements of $P_n\otimes P_n$ and $P_b$ of all its maximal
  elements. Nevertheless there are no arrows from one copy of $P_n$
  into the other. Hence $P_n\otimes P_n$ can only be written as a
  parallel composition of $P_n$ with itself. \qed
\end{proof}

\begin{corollary}
  $\mcal G_n\subset \mcal G_{n+1}$,
  $\single1_n\subset \single1_{n+1}$ and
  $\single1_\ast\not\subseteq \mcal G_n$ for any $n\ge 0$.
\end{corollary}

\begin{proof}
  $\mcal G_n\subseteq \mcal G_{n+1}$ and
  $\single1_n\subseteq \single1_{n+1}$ by definition, and by Lemma
  \ref{le:bintree},
  \begin{equation*}
    \mcal G_{n+1}\setminus \mcal G_n \cap \single1_{n+1}\setminus \mcal
    T_n = \single1_{n+1}\setminus \mcal G_n \supseteq \{P_{n+1}\}\ne
    \emptyset
  \end{equation*}
  for all $n$.  The last claim follows from the fact that
  $P_n\in \single1_\ast$ for all $n$. \qed
\end{proof}

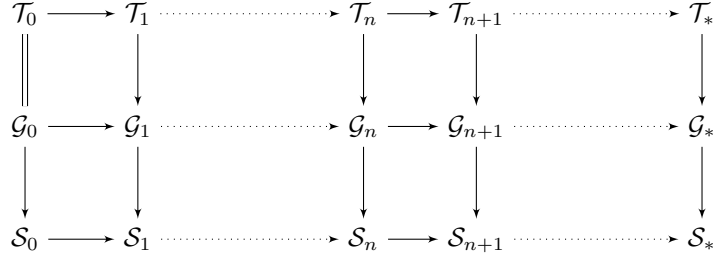
\begin{figure}[tbp]
  \hfill%
  \begin{tikzpicture}[x=1.5cm, y=1.5cm, shorten <=0pt, shorten >=0pt]
    \node (T0) at (0,0) {$\single1_0$};
    \node (T1) at (1,0) {$\single1_1$};
    \node (G0) at (0,-1) {$\mcal G_0$};
    \node (G1) at (1,-1) {$\mcal G_1$};
    \node (S0) at (0,-2) {$\single4_0$};
    \node (S1) at (1,-2) {$\single4_1$};
    \node (Tn) at (3,0) {$\single1_n$};
    \node (Tn+1) at (4,0) {$\vphantom{\single1_n}\smash{\single1_{n+1}}$};
    \node (Gn) at (3,-1) {$\mcal G_n$};
    \node (Gn+1) at (4,-1) {$\vphantom{\mcal G_n}\smash{\mcal G_{n+1}}$};
    \node (Sn) at (3,-2) {$\single4_n$};
    \node (Sn+1) at (4,-2) {$\vphantom{\single4_n}\smash{\single4_{n+1}}$};
    \node (Tinf) at (6,0) {$\vphantom{\single1_n}\smash{\single1_\ast}$};
    \node (Ginf) at (6,-1) {$\vphantom{\mcal G_n}\smash{\mcal G_\ast}$};
    \node (Sinf) at (6,-2) {$\vphantom{\single4_n}\smash{\single4_\ast}$};
    \path (T0) edge (T1);
    \path (G0) edge (G1);
    \path (S0) edge (S1);
    \path (Tn) edge (Tn+1);
    \path (Gn) edge (Gn+1);
    \path (Sn) edge (Sn+1);
    \path (T1) edge (G1);
    \path (G1) edge (S1);
    \path (Tn) edge (Gn);
    \path (Gn) edge (Sn);
    \path (Tn+1) edge (Gn+1);
    \path (Gn+1) edge (Sn+1);
    \path (Tinf) edge (Ginf);
    \path (Ginf) edge (Sinf);
    \path (T0) edge[-, double distance=1.5pt] (G0);
    \path (G0) edge (S0);
    \path (T1) edge[dotted] (Tn);
    \path (G1) edge[dotted] (Gn);
    \path (S1) edge[dotted] (Sn);
    \path (Tn+1) edge[dotted] (Tinf);
    \path (Gn+1) edge[dotted] (Ginf);
    \path (Sn+1) edge[dotted] (Sinf);
  \end{tikzpicture}
  \hfill\mbox{}
  \caption{Hasse diagram of equalities and proper inclusions between
    sp-posets (top row), gp-posets (middle) and gp-iposets (bottom
    row).}
  \label{fig:hierarchies}
\end{figure}

We summarise the relationships between the different sets in the next
proposition, see Figure \ref{fig:hierarchies} for an illustration.

\begin{proposition}
  \label{pr:sp-gp-gpi}
  For all $n\ge 0$, $\single1_n\subset \single1_{n+1}$,
  $\mcal G_n\subset \mcal G_{n+1}$ and
  $\single4_n\subset \single4_{n+1}$.  For all $n\ge 1$,
  $\single1_n\subset \mcal G_n\subset \single4_n$ and
  $\single1_0= \mcal G_0\subset \single4_0$.  For all $n,m\ge 1$,
  $\mcal G_n\not\subseteq \single1_m$ and
  $\single4_n\not\subseteq \mcal G_m$.
\end{proposition}

\begin{proof}
  We have already shown all but the third and the last two claims, and
  $\single4_n\subset \single4_{n+1}$ follows from the definition
  $\mcal G_n= \single4_n\cap \Pos$.
  Now if $\mcal G_n\subseteq \single1_m$ for some $n,m\ge 1$, then
  also $\mcal G_n\subseteq \single1_*$ in contradiction to
  $\single1_*\not\ni \N\in \mcal G_1$.  This also shows that
  $\mcal T_n\subset \mcal G_n$ for $n\geq 1$.  The last claim is clear
  as $\single4_0$ contains iposets with nonempty interfaces. \qed
\end{proof}

\section{Forbidden Substructures}

In this section we collect some combinatorial properties of
gluing-parallel posets and iposets and expose eleven forbidden
substructures.  We do not know whether our list of forbidden
substructures is comprehensive.

It is well-known, and follows directly from their characterisation as
being $\N$-free, that series-parallel posets are closed under induced
subposets.  We show that the same is true for gp-posets.
  
\begin{proposition}
  \label{pr:gp-closed-induced}
  All induced subposets of gp-posets are gluing-parallel.
\end{proposition}

\begin{proof}
  We prove by induction that removing one point from a gp-poset
  yields a gp-poset. The claim then follows by another induction.

  Our property is obviously true for the singleton poset.  Now let $P$
  be gluing-parallel and $x\in P$.  If $P$ has a non-trivial parallel
  decomposition $P=Q\otimes R$, then either $x\in Q$ or $x\in R$, and
  these cases are symmetric.  So suppose $x\in Q$. Then
  $Q\setminus\{x\}$ is gluing-parallel by the inductive hypothesis and
  $P\setminus\{x\}=(Q\setminus\{x\})\otimes R$.

  If $P$ has a non-trivial gluing decomposition $P=Q\glue R$, then
  either $x\in Q\setminus T_Q$, $x\in R\setminus S_R$, or
  $x\in T_Q=S_R$.  The first two cases can be handled just as above,
  given that $P\setminus\{x\}=(Q\setminus\{x\})\glue R$ respectively
  $P\setminus\{x\}=Q\glue(R\setminus\{x\})$.  For the last case,
  $P\setminus\{x\}=(Q\setminus\{x\})\glue(R\setminus\{x\})$. \qed
\end{proof}

By collecting all posets that are not gluing-parallel and weeding out
induced subposets, we obtain the following consequence of Proposition
\ref{pr:gp-closed-induced}.

\begin{corollary}
  There exists a set $\mcal F$ of posets such that for any
  $P, Q\in \mcal F$ with $P\ne Q$, neither $P$ nor $Q$ is an induced
  subposet of the other, and such that any poset is gluing-parallel if
  and only if it does not contain any element of $\mcal F$ as an
  induced subposet.
\end{corollary}

Such a set $\mcal F$ is said to consist of \emph{forbidden
  substructures}.  We have already mentioned that sp-posets have
precisely one forbidden substructure, the poset $\N$; also interval orders
have precisely one forbidden substructure, the poset $\twotwo$.  The
question is now whether gp-posets admit a \emph{finite} set of
forbidden substructures.

\begin{proposition}
  \label{pr:forbidden1}
  The following five posets are forbidden substructures for gp-posets:
  \begin{gather*}
    \NN =\! \vcenter
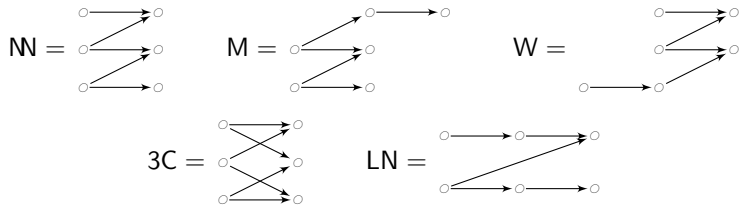
{\hbox{%
        \begin{tikzpicture}[y=.5cm]
          \node (0) at (0,0) {\intpt};
          \node (1) at (0,-1) {\intpt};
          \node (2) at (0,-2) {\intpt};
          \node (3) at (1,0) {\intpt};
          \node (4) at (1,-1) {\intpt};
          \node (5) at (1,-2) {\intpt};
          \path (0) edge (3) (1) edge (3) (1) edge (4) (2) edge (4)
          (2) edge (5);
        \end{tikzpicture}}}
    \qquad%
    \NPLUS =\! \vcenter{\hbox{%
        \begin{tikzpicture}[y=.5cm]
          \node (1) at (0,-1) {\intpt};
          \node (2) at (0,-2) {\intpt};
          \node (3) at (1,0) {\intpt};
          \node (4) at (1,-1) {\intpt};
          \node (5) at (1,-2) {\intpt};
          \node (6) at (2,0) {\intpt};
          \path (3) edge (6) (1) edge (3) (1) edge (4) (2) edge (4)
          (2) edge (5);
        \end{tikzpicture}}}
    \qquad%
    \NMINUS =\! \vcenter{\hbox{%
        \begin{tikzpicture}[y=.5cm]
          \node (0) at (0,0) {\intpt};
          \node (1) at (0,-1) {\intpt};
          \node (2) at (0,-2) {\intpt};
          \node (3) at (1,0) {\intpt};
          \node (4) at (1,-1) {\intpt};
          \node (-1) at (-1,-2) {\intpt};
          \path (0) edge (3) (1) edge (3) (1) edge (4) (2) edge (4)
          (-1) edge (2);
        \end{tikzpicture}}}
    \\
    \TC =\! \vcenter{\hbox{%
        \begin{tikzpicture}[y=.5cm]
          \node (0) at (0,0) {\intpt};
          \node (1) at (0,-1) {\intpt};
          \node (2) at (0,-2) {\intpt};
          \node (3) at (1,0) {\intpt};
          \node (4) at (1,-1) {\intpt};
          \node (5) at (1,-2) {\intpt};
          \path (0) edge (3) (1) edge (3) (1) edge (5) (2) edge (4)
          (2) edge (5) (0) edge (4);
        \end{tikzpicture}}}
    \qquad%
    \LN =\! \vcenter{\hbox{%
        \begin{tikzpicture}[y=.7cm]
          \node (0) at (0,0) {\intpt};
          \node (1) at (0,-1) {\intpt};
          \node (2) at (1,0) {\intpt};
          \node (3) at (1,-1) {\intpt};
          \node (4) at (2,0) {\intpt};
          \node (5) at (2,-1) {\intpt};
          \path (0) edge (2) (2) edge (4) (1) edge (4) (1) edge (3)
          (3) edge (5);
        \end{tikzpicture}}}
  \end{gather*}
\end{proposition}

\begin{proof}
  $\NMINUS$ is the opposite of $\NPLUS$, so we can ignore it for this
  proof.  We use the notation introduced below Lemma
  \ref{le:characteristic-function} and label the vertices of the
  remaining four posets so that elements of $P_a$ are labelled $a_i$,
  elements of $P_b$ are labelled $b_i$, and the remaining elements are
  labelled $p_i$:
  \begin{gather*}
    \NN =\! \vcenter{\hbox{%
        \begin{tikzpicture}[y=.7cm, x=1.2cm]
          \node (0) at (0,0) {$\vphantom{bp}p_1$};
          \node (1) at (0,-1) {$\vphantom{bp}a_1$};
          \node (2) at (0,-2) {$\vphantom{bp}a_2$};
          \node (3) at (1,0) {$\vphantom{bp}b_2$};
          \node (4) at (1,-1) {$\vphantom{bp}b_1$};
          \node (5) at (1,-2) {$\vphantom{bp}p_2$};
          \path (0) edge (3) (1) edge (3) (1) edge (4) (2) edge (4)
          (2) edge (5);
        \end{tikzpicture}}}
    \qquad\quad%
    \NPLUS =\! \vcenter{\hbox{%
        \begin{tikzpicture}[y=.7cm, x=1.2cm]
          \node (1) at (0,-1) {$\vphantom{bp}a$};
          \node (2) at (0,-2) {$\vphantom{bp}p_1$};
          \node (3) at (1,0) {$\vphantom{bp}p_2$};
          \node (4) at (1,-1) {$\vphantom{bp}b_2$};
          \node (5) at (1,-2) {$\vphantom{bp}p_3$};
          \node (6) at (2,0) {$\vphantom{bp}b_1$};
          \path (3) edge (6) (1) edge (3) (1) edge (4) (2) edge (4)
          (2) edge (5);
        \end{tikzpicture}}}
    \qquad%
    \TC =\! \vcenter{\hbox{%
        \begin{tikzpicture}[y=.7cm, x=1.2cm]
          \node (0) at (0,0) {$\vphantom{bp}a_1$};
          \node (1) at (0,-1) {$\vphantom{bp}a_2$};
          \node (2) at (0,-2) {$\vphantom{bp}a_3$};
          \node (3) at (1,0) {$\vphantom{bp}b_1$};
          \node (4) at (1,-1) {$\vphantom{bp}b_2$};
          \node (5) at (1,-2) {$\vphantom{bp}b_3$};
          \path (0) edge (3) (1) edge (3) (1) edge (5) (2) edge (4)
          (2) edge (5) (0) edge (4);
        \end{tikzpicture}}}
    \\[1ex]
    \LN =\! \vcenter{\hbox{%
        \begin{tikzpicture}[y=.8cm, x=1.0cm]
          \node (0) at (0,0) {$\vphantom{bp}p_1$};
          \node (1) at (0,-1) {$\vphantom{bp}a$};
          \node (2) at (1,0) {$\vphantom{bp}p_2$};
          \node (3) at (1,-1) {$\vphantom{bp}p_3$};
          \node (4) at (2,0) {$\vphantom{bp}b$};
          \node (5) at (2,-1) {$\vphantom{bp}p_4$};
          \path (0) edge (2) (2) edge (4) (1) edge (4) (1) edge (3)
          (3) edge (5);
        \end{tikzpicture}}}
  \end{gather*}

  Let $P\in\{\NN, \NPLUS, \TC, \LN\}$ and assume that $P$ has a
  non-trivial gluing decomposition.  Lemma
  \ref{le:criterion-for-gluing} implies that there exists a function
  $\phi:P\to\{0,*,1\}$ satisfying
  Lemma \ref{le:characteristic-function} such that $\phi(a_i)=1$ and
  $\phi(b_i)=0$.

  For $P=\NN$ and $P=\TC$ we have $a_2\not<b_2$ in contradiction to
  $\phi(a_2)=1$ and $\phi(b_2)=0$.  For $P=\NPLUS$, $a\not<p_1\not<b_1$
  implies $\phi(p_1)\notin\{0,1\}$ by the same argument, thus
  $\phi(p_1)=*$. Similarly, $a\not<p_3\not<b_1$ implies $\phi(p_3)=*$.
  This contradicts $p_1<p_3$ by Lemma
  \ref{le:characteristic-function}\ref{le:characteristic-functionC}.

  For $P=\LN$, $a\not<p_2$ implies $\phi(p_2)\neq 0$, \ie
  $\phi(p_2)\in\{*, 1\}$, and then $p_1<p_2$ implies $\phi(p_1)=1$ by
  Lemma
  \ref{le:characteristic-function}\ref{le:characteristic-functionD}.
  Dually, $p_3\not<b$ implies $\phi(p_3)\neq 1$ and then
  $\phi(p_4)=0$.  But $p_1\not< p_4$, a contradiction.

  The proof that all proper subposets of these posets are
  gluing-parallel is trivial verification. \qed
\end{proof}

In addition to the above five posets, we have found six other
elements of $\mcal F$.  Whether $\mcal F$ is finite or infinite is
open.

\begin{figure}[tbp]
  \centering
  \begin{tikzpicture}[x=1.2cm, y=.7cm]
    \begin{scope}[shift={(0,0)}]
      \node (2) at (0,-1.3) {\intpt};
      \node (0) at (-.8,-2) {\intpt};
      \node (1) at (0,-2.7) {\intpt};
      \node (4) at (.5,0)  {\intpt};
      \node (6) at (1,-1.3) {\intpt};
      \node (3) at (1.3,-2) {\intpt};
      \node (7) at (1,-2.7) {\intpt};
      \node (5) at (.5,-4) {\intpt};
      \node (8) at (2,0)  {\intpt};
      \node (9) at (2,-4) {\intpt};
      \path (0) edge (4) (0) edge (5) (1) edge (3) (1) edge (6) (1) edge
      (9) (2) edge (3) (2) edge (7) (2) edge (8) (4) edge (8) (5) edge
      (9) (6) edge (8) (7) edge (9);
    \end{scope}
    \begin{scope}[shift={(3,-.5)}]
      \node (0) at (0,0)  {\intpt};
      \node (1) at (0,-3) {\intpt};
      \node (2) at (1,0)  {\intpt};
      \node (3) at (1,-1) {\intpt};
      \node (4) at (1,-2) {\intpt};
      \node (5) at (1,-3) {\intpt};
      \node (6) at (2,0)  {\intpt};
      \node (7) at (2,-3) {\intpt};
      \path (0) edge (2) (0) edge (3) (1) edge (4) (1) edge (5) (2)
      edge (6) (3) edge (7) (4) edge (6) (5) edge (7);
    \end{scope}
    \begin{scope}[shift={(6,0)}]
      \node (2) at (0,0)  {\intpt};
      \node (1) at (0,-4) {\intpt};
      \node (0) at (.7,-2) {\intpt};
      \node (5) at (1,-1.3) {\intpt};
      \node (4) at (1,-2.7) {\intpt};
      \node (6) at (1.5,0)  {\intpt};
      \node (3) at (1.5,-4) {\intpt};
      \node (7) at (2,-1.3) {\intpt};
      \node (8) at (2,-2.7) {\intpt};
      \node (9) at (2.8,-2) {\intpt};
      \path (0) edge (7) (0) edge (8) (1) edge (3) (1) edge (4) (1) edge
      (8) (2) edge (5) (2) edge (6) (2) edge (7) (3) edge (9) (4) edge
      (7) (5) edge (8) (6) edge (9);
    \end{scope}
    \begin{scope}[shift={(-.5,-5)}]
      \node (0) at (0,-1) {\intpt};
      \node (1) at (0,-3) {\intpt};
      \node (2) at (1,-2) {\intpt};
      \node (3) at (1,-1) {\intpt};
      \node (4) at (1,0)  {\intpt};
      \node (5) at (1,-4) {\intpt};
      \node (6) at (1,-3) {\intpt};
      \node (7) at (2,-1) {\intpt};
      \node (8) at (2,-3) {\intpt};
      \node (9) at (2,-2) {\intpt};
      \path (0) edge (2) (0) edge (3) (0) edge (4) (1) edge (2) (1)
      edge (5) (1) edge (6) (2) edge (7) (2) edge (8) (3) edge (9) (4)
      edge (7) (5) edge (8) (6) edge (9);
    \end{scope}
    \begin{scope}[shift={(2.5,-4.5)}]
      \node (0) at (0,-1) {\intpt};
      \node (1) at (0,-4) {\intpt};
      \node (2) at (1.5,-4) {\intpt};
      \node (3) at (1.5,-1) {\intpt};
      \node (4) at (1,-2) {\intpt};
      \node (5) at (1,-3) {\intpt};
      \node (6) at (2,-3) {\intpt};
      \node (7) at (2,-2) {\intpt};
      \node (8) at (3,-4) {\intpt};
      \node (9) at (3,-1) {\intpt};
      \path (0) edge (3) (0) edge (4) (0) edge (7) (1) edge (2) (1) edge
      (5) (1) edge (6) (2) edge (8) (3) edge (9) (4) edge (6) (4) edge
      (9) (5) edge (7) (5) edge (8) (6) edge (8) (7) edge (9);
    \end{scope}
    \begin{scope}[shift={(6.5,-5)}]
      \node (0) at (0,-1) {\intpt};
      \node (1) at (0,-3) {\intpt};
      \node (2) at (0,-2) {\intpt};
      \node (3) at (1,-2) {\intpt};
      \node (4) at (1,0)  {\intpt};
      \node (5) at (1,-4) {\intpt};
      \node (6) at (1,-1) {\intpt};
      \node (7) at (1,-3) {\intpt};
      \node (8) at (2,-1) {\intpt};
      \node (9) at (2,-3) {\intpt};
      \path (0) edge (3) (0) edge (4) (1) edge (3) (1) edge (5) (2) edge
      (6) (2) edge (7) (3) edge (8) (3) edge (9) (4) edge (8) (5) edge
      (9) (6) edge (8) (7) edge (9);
    \end{scope}
  \end{tikzpicture}
  \caption{Additional forbidden substructures for gp-posets.}
  \label{fi:forbidden2}
\end{figure}
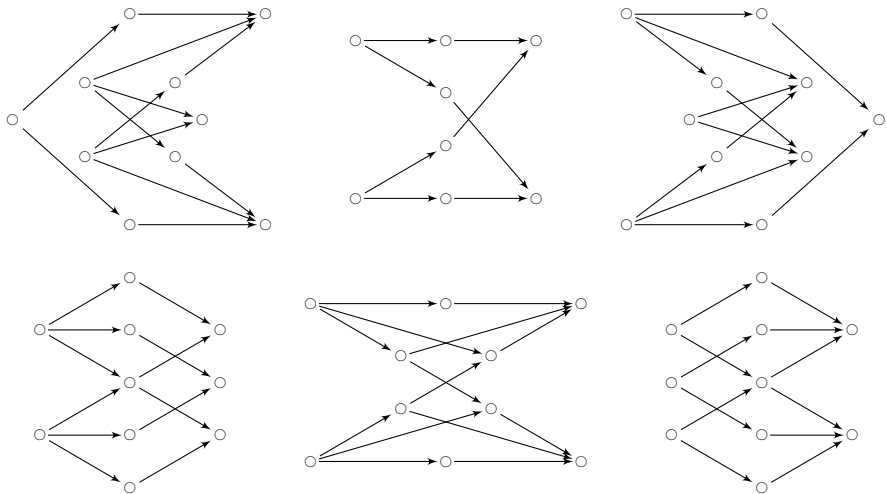

\begin{proposition}
  \label{pr:forbidden2}
  The six posets in Figure \ref{fi:forbidden2} are forbidden
  substructures for gp-posets.
\end{proposition}

\begin{proof}
  By computer (see below).
\end{proof}

\begin{remark}
  The 8-point forbidden substructure in the centre of the top row of
  Figure \ref{fi:forbidden2} has a gluing decomposition along a maximal
  antichain:
  
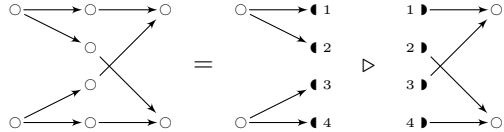
\begin{equation*}
    \begin{tikzpicture}[y=.5cm]
      \begin{scope}
        \node (0) at (0,0)  {\intpt};
        \node (1) at (0,-3) {\intpt};
        \node (2) at (1,0)  {\intpt};
        \node (3) at (1,-1) {\intpt};
        \node (4) at (1,-2) {\intpt};
        \node (5) at (1,-3) {\intpt};
        \node (6) at (2,0)  {\intpt};
        \node (7) at (2,-3) {\intpt};
        \path (0) edge (2) (0) edge (3) (1) edge (4) (1) edge (5) (2)
        edge (6) (3) edge (7) (4) edge (6) (5) edge (7);
        \node at (2.5,-1.5) {$=$};
      \end{scope}
      \begin{scope}[shift={(3,0)}]
        \node (0) at (0,0)  {\intpt};
        \node (1) at (0,-3) {\intpt};
        \node [label=right:{\tiny $1$}] (2) at (1,0)  {\outpt};
        \node [label=right:{\tiny $2$}] (3) at (1,-1) {\outpt};
        \node [label=right:{\tiny $3$}] (4) at (1,-2) {\outpt};
        \node [label=right:{\tiny $4$}] (5) at (1,-3) {\outpt};
        \path (0) edge (2) (0) edge (3) (1) edge (4) (1) edge (5);
        \node at (1.7,-1.5) {$\glue$};
      \end{scope}
      \begin{scope}[shift={(4.4,0)}]
        \node [label=left:{\tiny $1$}] (2) at (1,0)  {\inpt};
        \node [label=left:{\tiny $2$}] (3) at (1,-1) {\inpt};
        \node [label=left:{\tiny $3$}] (4) at (1,-2) {\inpt};
        \node [label=left:{\tiny $4$}] (5) at (1,-3) {\inpt};
        \node (6) at (2,0)  {\intpt};
        \node (7) at (2,-3) {\intpt};
        \path (2) edge (6) (3) edge (7) (4) edge (6) (5) edge (7);
      \end{scope}
    \end{tikzpicture}
  \end{equation*}
  Now the first of the iposets on the right-hand side is easily seen
  to be gluing-parallel, given that it is the parallel product of two
  three-point iposets.  The second iposet, however, has its interfaces
  swapped so that it is \emph{not} a parallel product; it is not
  gluing-parallel.
\end{remark}

We have written a Julia program, using the LightGraphs package
\cite{BrombergerF17}, to generate gluing-parallel iposets and analyse
their properties.\footnote{Our Julia code is available at
  \url{https://github.com/ulifahrenberg/pomsetproject/tree/main/code/20210618/},
  and the data at
  \url{https://github.com/ulifahrenberg/pomsetproject/tree/main/data/}.}
The next proposition is a result of these calculations.

\begin{proposition}
  The eleven posets of Propositions \ref{pr:forbidden1} and
  \ref{pr:forbidden2} are the only posets in $\mcal F$ with at most 10
  points. \qed
\end{proposition}

\begin{table}[tbp]
  \centering
  \caption{Different types of posets and iposets on $n$ points: all
    posets; sp-posets; interval orders; gp-posets; iposets;
    gp-iposets.  The last line refers to the column sequence's index
    in Neil Sloane's Encyclopedia of Integer Sequences,
    \url{http://oeis.org/}, if available.}
  \label{ta:numposets}
  \smallskip
  \begin{tabular}{r|rrrrrrrr}
    $n$ & $\textsf{P}(n)$ & $\textsf{SP}(n)$ & $\textsf{IO}(n)$ &
    $\textsf{GP}(n)$ & $\textsf{IP}(n)$ & $\textsf{GPI}(n)$ \\\hline
    0 & 1 & 1 & 1 & 1 & 1 & 1 \\
    1 & 1 & 1 & 1 & 1 & 4 & 4 \\
    2 & 2 & 2 & 2 & 2 & 17 & 16 \\
    3 & 5 & 5 & 5 & 5 & 86 & 74 \\
    4 & 16 & 15 & 15 & 16 & 532 & 419 \\
    5 & 63 & 48 & 53 & 63 & 4068 & 2980 \\
    6 & 318 & 167 & 217 & 313 & 38.933 & 26.566 \\
    7 & 2045 & 602 & 1014 & 1903 & 474.822 &
    289.279 \\
    8 & 16.999 & 2256 & 5335 & 13.943 & 7.558.620 &
    3.726.311 \\
    9 & 183.231 & 8660 & 31.240 & 120.442 & & \\
    10 & 2.567.284 & 33.958 & 201.608 & 1.206.459 & & \\
    11 & 46.749.427 & 135.292 & 1.422.074 \\[1ex]
    EIS & 112 & 3430 & 22493 & & 
  \end{tabular}
\end{table}

In other words, any further forbidden substructures must have at least
11 points.  Generating posets is notoriously difficult
\cite{DBLP:journals/order/BrinkmannM02}, so any improvements to the
above results are left for future work.  We have also used our
software to count non-isomorphic posets and iposets of different
types, see Table \ref{ta:numposets}.  As a refutation of a conjecture
in \cite{DBLP:conf/RelMiCS/FahrenbergJST20}, the sequence
$\textsf{GP}(n)$ is not equal to EIS sequence 79566, and there appears
to be no relation between gp-posets and $C_4$-free connected graphs.

We comment on some details in Table~\ref{ta:numposets}. The
difference $\textsf{P}(4) - \textsf{SP}(4) = 1$ is witnessed by the
$\N$ poset that is not series-parallel, whereas 
$\textsf{P}(4) - \textsf{IO}(4) = 1$ is witnessed by the $\twotwo$
poset. The difference $\textsf{P}(6) - \textsf{GP}(6) = 5$ is
witnessed by the structures in Proposition \ref{pr:forbidden1}. The
differences between $\textsf{IP}(n)$ and $\textsf{GPI}(n)$ for $n\le
4$
have been discussed in Example \ref{ex:non-gp-posets-parallel}.

\section{Conclusion}

We have introduced posets with interfaces (iposets) with a new gluing
operation that generalises the standard serial composition of posets,
but identifies some of the maximal points of the first poset with some
of the minimal points of the second.  In the interpretation of posets
as components of concurrent events, such interfaces allow events to
continue across components or, alternatively, decompositions of posets
with respect to synchronic cuts in time.

The idea of equipping posets or pomsets with interfaces in concurrency
theory can be traced back to Winkowski
\cite{DBLP:journals/ipl/Winkowski77}, who introduces a gluing
composition for pomsets without autoconcurrency where all maximal
elements of the first pomset are merged with the corresponding minimal
elements of the second.  Our operation generalises this to interfaces
consisting of subsets of maximal or minimal elements and gluing
composition guided by interface identification rather than labels.
Similar interface-based compositions of graphs and posets have been
explored by Courcelle and Engelfriet \cite{DBLP:books/daglib/0030804},
Fiore and Devesas Campos \cite{DBLP:conf/birthday/FioreC13}, and
Mimram \cite{DBLP:journals/corr/Mimram15}, but these do not introduce
precedence when gluing and give rise to simpler algebraic structures,
monoidal categories and PROPs, with strict interchange laws.  Further,
\cite{DBLP:conf/birthday/FioreC13, DBLP:journals/corr/Mimram15} remove
interfaces when gluing; whereas in our case, thinking of points as
events in concurrent systems, we should rather keep them.

We have shown that iposets under gluing and parallel composition and
with a suitable notion of subsumption form a 2-category with lax
tensor.  This generalises the concurrent monoids used in freeness
results for series-parallel pomsets.  Whether the iposets generated
from singletons by finitary gluing and parallel compositions form the
free 2-category with some adaptations remains open.

We have also shown that the hierarchy of gluing-parallel iposets
defined by the alternation levels of the two compositions is infinite.
This supports the idea that the algebra of these gp-iposets should be
freely generated in one way or another.  Further, our hierarchy lies
above that for series-parallel posets and captures interval orders at
its second alternation level.

Using a computer program, we have found five posets on 6 points, one
on 8, and another five on 10 points, which are forbidden induced
substructures for gp-posets.  The five substructures on 6 points
cannot be decomposed along interfaces, whereas the other six posets
above are decomposable, but their components have their interfaces
swapped.  If such swapping were permitted for gp-iposets, then the six
bigger forbidden substructures would disappear.  Exploring the
resulting \emph{gluing-parallel-symmetric} iposets is left for future
work.

We have recently defined languages of higher-dimensional automata as
sets of ipomsets whose underlying posets are interval orders
\cite{Hdalang}.  We conjecture that such automata are the machine
model for ipomset languages, but leave this exploration to future
work. A final avenue for future work is the development of a higher
dimensional Kleene algebra for the languages for these automata and
more general for languages of iposets that arise from lifting their
2-category to the powerset level.

\section*{Acknowledgement}

We are grateful to James Cranch and Samuel Mimram for numerous
discussions and clarifications on the subjects of this work, and to
the anonymous reviewers for helping us to greatly improve this paper.
The first author acknowledges the support of the \textit{Chaire ISC~:
  Engineering Complex Systems} and École polytechnique, where most of
this work was carried out.  The third author acknowledges financial
support from EPSRC, grant EP/R032351/1, an invited professorship at
\emph{Labex DigiCosme}, École polytechnique and a fellowship at
\emph{Collegium de Lyon} during the final stages of this work.

\newcommand{\Afirst}[1]{#1} \newcommand{\afirst}[1]{#1}


\begin{thebibliography}{10}

\bibitem{DBLP:journals/tcs/BloomE96a}
Stephen~L. Bloom and Zolt{\'{a}}n {\'{E}}sik.
\newblock Free shuffle algebras in language varieties.
\newblock {\em {Theoretical Computer Science}}, 163(1{\&}2):55--98, 1996.

\bibitem{DBLP:journals/mscs/BloomE97}
Stephen~L. Bloom and Zolt{\'{a}}n {\'{E}}sik.
\newblock Varieties generated by languages with poset operations.
\newblock {\em {Mathematical Structures in Computer Science}}, 7(6):701--713,
  1997.

\bibitem{book/BridsonH99}
Martin~R. Bridson and Andr{\'e} Haefliger.
\newblock {\em Metric Spaces of Non-positive Curvature}, volume 319 of {\em
  Grundlehren der mathematischen Wissenschaften}.
\newblock {Springer}, 1999.

\bibitem{DBLP:journals/order/BrinkmannM02}
Gunnar Brinkmann and Brendan~D. McKay.
\newblock Posets on up to 16 points.
\newblock {\em Order}, 19(2):147--179, 2002.

\bibitem{BrombergerF17}
Seth Bromberger, James Fairbanks, et~al.
\newblock \texttt{JuliaGraphs/LightGraphs.jl}: an optimized graphs package for
  the {Julia} programming language, 2017.
\newblock \url{https://doi.org/10.5281/zenodo.889971}.

\bibitem{DBLP:books/daglib/0030804}
Bruno Courcelle and Joost Engelfriet.
\newblock {\em Graph Structure and Monadic Second-Order Logic - {A}
  Language-Theoretic Approach}.
\newblock Cambridge University Press, 2012.

\bibitem{DaveyP12}
Brian~A. Davey and Hilary~A. Priestley.
\newblock {\em Introduction to Lattices and Order}.
\newblock Cambridge University Press, 2012.

\bibitem{journals/tcs/Esik02}
Zolt{\'{a}}n {\'{E}}sik.
\newblock Axiomatizing the subsumption and subword preorders on finite and
  infinite partial words.
\newblock {\em {Theoretical Computer Science}}, 273(1-2):225--248, 2002.

\bibitem{DBLP:conf/RelMiCS/FahrenbergJST20}
Uli Fahrenberg, Christian Johansen, Georg Struth, and Ratan~Bahadur Thapa.
\newblock Generating posets beyond {N}.
\newblock In Uli Fahrenberg, Peter Jipsen, and Michael Winter, editors, {\em
  RAMiCS}, volume 12062 of {\em {Lecture Notes in Computer Science}}, pages
  82--99, Heidelberg, 2020. {Springer}.

\bibitem{Hdalang}
Uli Fahrenberg, Christian Johansen, Georg Struth, and Krzysztof Ziemia{\'n}ski.
\newblock Languages of higher-dimensional automata.
\newblock {\em {Mathematical Structures in Computer Science}}, 31(5):575--613,
  2021.

\bibitem{DBLP:conf/birthday/FioreC13}
Marcelo~P. Fiore and Marco Devesas~Campos.
\newblock The algebra of directed acyclic graphs.
\newblock In Bob Coecke, Luke Ong, and Prakash Panangaden, editors, {\em
  Computation, Logic, Games, and Quantum Foundations}, volume 7860 of {\em
  {Lecture Notes in Computer Science}}, pages 37--51. {Springer}, 2013.

\bibitem{journals/mpsy/Fishburn70}
Peter~C. Fishburn.
\newblock Intransitive indifference with unequal indifference intervals.
\newblock {\em {Journal of Mathematical Psychology}}, 7(1):144--149, 1970.

\bibitem{book/Fishburn85}
Peter~C. Fishburn.
\newblock {\em Interval Orders and Interval Graphs: A Study of Partially
  Ordered Sets}.
\newblock Wiley, 1985.

\bibitem{DBLP:journals/tcs/Gischer88}
Jay~L. Gischer.
\newblock The equational theory of pomsets.
\newblock {\em {Theoretical Computer Science}}, 61:199--224, 1988.

\bibitem{DBLP:journals/fuin/Grabowski81}
Jan Grabowski.
\newblock On partial languages.
\newblock {\em {Fundamentae Informatica}}, 4(2):427, 1981.

\bibitem{journals/tac/Haucourt06}
Emmanuel Haucourt.
\newblock Categories of components and loop-free categories.
\newblock {\em {Theory Appl. Categ.}}, 16(27):736--770, 2006.

\bibitem{DBLP:journals/toplas/HerlihyW90}
Maurice Herlihy and Jeannette~M. Wing.
\newblock Linearizability: {A} correctness condition for concurrent objects.
\newblock {\em {{ACM} Transactions on Programming Languages and Systems}},
  12(3):463--492, 1990.

\bibitem{DBLP:journals/jlp/HoareMSW11}
Tony Hoare, Bernhard M{\"{o}}ller, Georg Struth, and Ian Wehrman.
\newblock Concurrent {K}leene algebra and its foundations.
\newblock {\em {Journal of Logic and Algebraic Programming}}, 80(6):266--296,
  2011.

\bibitem{DBLP:conf/apn/Janicki18}
Ryszard Janicki.
\newblock Modeling operational semantics with interval orders represented by
  sequences of antichains.
\newblock In Victor Khomenko and Olivier~H. Roux, editors, {\em PETRI NETS},
  volume 10877 of {\em {Lecture Notes in Computer Science}}, pages 251--271.
  {Springer}, 2018.

\bibitem{DBLP:journals/iandc/JanickiY17}
Ryszard Janicki and Xiang Yin.
\newblock Modeling concurrency with interval traces.
\newblock {\em {Information and Computation}}, 253:78--108, 2017.

\bibitem{book/JohnsonY21}
Niles Johnson and Donald Yau.
\newblock {\em 2-Dimensional Categories}.
\newblock Oxford University Press, 2021.
\newblock \url{https://arxiv.org/abs/2002.06055}.

\bibitem{DBLP:journals/tcs/LodayaW00}
Kamal Lodaya and Pascal Weil.
\newblock Series-parallel languages and the bounded-width property.
\newblock {\em {Theoretical Computer Science}}, 237(1-2):347--380, 2000.

\bibitem{conf/pda/Mattern88}
Friedemann Mattern.
\newblock Virtual time and global states of distributed systems.
\newblock In Michel Cosnard, editor, {\em Workshop on Parallel and Distributed
  Algorithms}, pages 215--226. Elsevier, 1989.

\bibitem{DBLP:journals/corr/Mimram15}
Samuel Mimram.
\newblock Presenting finite posets.
\newblock In Aart Middeldorp and Femke van Raamsdonk, editors, {\em TERMGRAPH},
  volume 183 of {\em {Electronic Proceedings in Theoretical Computer Science}},
  pages 1--17, 2014.

\bibitem{Pratt86pomsets}
Vaughan~R. Pratt.
\newblock Modeling concurrency with partial orders.
\newblock {\em Journal of Parallel Programming}, 15(1):33--71, 1986.

\bibitem{Pratt03trans_cancel}
Vaughan~R. Pratt.
\newblock Transition and cancellation in concurrency and branching time.
\newblock {\em {Mathematical Structures in Computer Science}}, 13(4):485--529,
  2003.

\bibitem{DBLP:books/mc/21/Struth21}
Georg Struth.
\newblock Trimming the hedges: An algebra to tame concurrency.
\newblock In Cliff~B. Jones and Jayadev Misra, editors, {\em Theories of
  Programming: The Life and Works of Tony Hoare}, pages 317--346. {ACM} /
  Morgan {\&} Claypool, 2021.

\bibitem{DBLP:conf/ifip2/Glabbeek90}
Rob~J. van Glabbeek.
\newblock The refinement theorem for {ST}-bisimulation semantics.
\newblock In Manfred Broy and Cliff~B. Jones, editors, {\em IFIP TC2 Working
  Conference on Programming Concepts and Methods}, pages 27--52. North-Holland,
  1990.

\bibitem{DBLP:journals/tcs/Glabbeek06}
Rob~J. van Glabbeek.
\newblock \afirst{On} the expressiveness of higher dimensional automata.
\newblock {\em {Theoretical Computer Science}}, 356(3):265--290, 2006.
\newblock See also~\cite{DBLP:journals/tcs/Glabbeek06a}.

\bibitem{DBLP:journals/tcs/Glabbeek06a}
Rob~J. van Glabbeek.
\newblock Erratum to ``{O}n the expressiveness of higher dimensional
  automata''.
\newblock {\em {Theoretical Computer Science}}, 368(1-2):168--194, 2006.

\bibitem{DBLP:journals/acta/GlabbeekG01}
Rob~J. van Glabbeek and Ursula Goltz.
\newblock Refinement of actions and equivalence notions for concurrent systems.
\newblock {\em Acta Informatica}, 37(4/5):229--327, 2001.

\bibitem{DBLP:books/sp/Vogler92}
Walter Vogler.
\newblock {\em Modular Construction and Partial Order Semantics of Petri Nets},
  volume 625 of {\em {Lecture Notes in Computer Science}}.
\newblock {Springer}, 1992.

\bibitem{Wiener14}
Norbert Wiener.
\newblock A contribution to the theory of relative position.
\newblock {\em Proceedings of the Cambridge Philosophical Society},
  17:441--449, 1914.

\bibitem{DBLP:journals/ipl/Winkowski77}
J{\'{o}}zef Winkowski.
\newblock An algebraic characterization of the behaviour of non-sequential
  systems.
\newblock {\em Information Processing Letters}, 6(4):105--109, 1977.
\end{thebibliography}
\end{document}